\documentclass[a4paper,10pt]{article}
\usepackage[utf8]{inputenc}

\usepackage{palatino} 

\usepackage{fullpage}
\usepackage[linesnumbered,ruled]{algorithm2e}
\usepackage[dvipsnames]{xcolor}
\usepackage{todonotes}
\usepackage{xspace}
\usepackage{complexity}
\usepackage{microtype}
\usepackage{tikz}
\usepackage{amsmath}
\usepackage{amsthm}
\usepackage{amssymb}
\usepackage{doi}
\usepackage[affil-it]{authblk}
\usepackage{subcaption}
\usepackage{hyperref}

\newtheorem{theorem}{Theorem}

\newtheorem{lemma}[theorem]{Lemma}

\newtheorem{corollary}[theorem]{Corollary}
\newtheorem{remark}{Remark}
\newtheorem{claim}{Claim}

\newcommand*\samethanks[1][\value{footnote}]{\footnotemark[#1]}

\newcommand{\TAR}{{\sf TAR}}
\newcommand{\claimqed}{\hfill $\Diamond$ \vspace{1em}}
\newcommand{\STWML}{\textsc{Spanning Tree with Many Leaves}\xspace}
\newcommand{\STWkL}{\textsc{Spanning Tree with At Most $k$ Leaves}\xspace}
\newcommand{\STWThL}{\textsc{Spanning Tree with At Most $3$ Leaves}\xspace}

\newcommand{\source}{{\sf s}}
\newcommand{\target}{{\sf t}}
\newcommand{\canonical}{{\sf c}}
\renewcommand{\o}{\text{o}}

\newenvironment{proofclaim}[1][]%
    {\noindent \emph{Proof.} {}{#1}{}}{}

\title{Reconfiguration of Spanning Trees with Many or Few Leaves \thanks{Partially supported by JSPS and MEAE-MESRI under the Japan-France Integrated Action Program (SAKURA)}}

\author[1]{Nicolas Bousquet \thanks{Partially supported by ANR project GrR (ANR-18-CE40-0032).}}
\author[2]{Takehiro Ito \thanks{Partially supported by JSPS KAKENHI Grant Numbers JP18H04091 and JP19K11814, Japan.}}
\author[3]{Yusuke Kobayashi \thanks{Supported by JSPS KAKENHI Grant Numbers 17K19960, 18H05291, and JP20K11692, Japan.}}
\author[2]{Haruka Mizuta}
\author[4]{Paul Ouvrard \samethanks[2]}
\author[2]{Akira Suzuki \thanks{Partially supported by JSPS KAKENHI Grant Numbers JP18H04091 and JP20K11666, Japan.}}
\author[5]{Kunihiro Wasa \thanks{Partially supported by JST CREST Grant Numbers JPMJCR18K3 and JPMJCR1401, and JSPS KAKENHI Grant Number 19K20350, Japan.}}

\affil[1]{CNRS, LIRIS, Universit\'e de Lyon, Universit\'e Claude Bernard Lyon 1, Lyon, France}
\affil[2]{Graduate School of Information Sciences, Tohoku University, Japan}
\affil[3]{Research Institute for Mathematical Sciences, Kyoto University, Japan}
\affil[4]{Univ. Bordeaux, Bordeaux INP, CNRS, LaBRI, UMR5800, F-33400 Talence, France}
\affil[5]{Toyohashi University of Technology, Japan}

\date{}

\begin{document}
\maketitle

\begin{abstract}
Let $G$ be a graph and $T_1,T_2$ be two spanning trees of $G$. We say that $T_1$ can be transformed into $T_2$ via an edge flip if there exist two edges $e \in T_1$ and $f$ in $T_2$ such that $T_2= (T_1 \setminus e) \cup f$. Since spanning trees form a matroid, one can indeed transform a spanning tree into any other via a sequence of edge flips, as observed in~\cite{Ito11}.

We investigate the problem of determining, given two spanning trees $T_1,T_2$ with an additional property $\Pi$, if there exists an edge flip transformation from $T_1$ to $T_2$ keeping property $\Pi$ all along.

First we show that determining if there exists a transformation from $T_1$ to $T_2$ such that all the trees of the sequence have at most $k$ (for any fixed $k \ge 3$) leaves is PSPACE-complete.

We then prove that determining if there exists a transformation from $T_1$ to $T_2$ such that all the trees of the sequence have at least $k$ leaves (where $k$ is part of the input) is PSPACE-complete even restricted to split, bipartite or planar graphs. We complete this result by showing that the problem becomes polynomial for cographs, interval graphs and when $k=n-2$.
\end{abstract}

\section{Introduction}

Given an instance of some combinatorial search problem and two of its
feasible solutions, a \emph{reconfiguration problem} asks whether one solution can be transformed into the other in a step-by-step fashion, such that each intermediate solution is also feasible.
Reconfiguration problems capture dynamic situations, where some
solution is in place and we would like to move to a desired alternative
solution without becoming infeasible. A systematic study of the complexity of reconfiguration problems was initiated in~\cite{Ito11}. Recently the topic has gained a lot of attention in the context of constraint satisfaction problems and graph problems, such as the independent set problem, the matching problem, and the dominating set problem. Reconfiguration problems naturally arise for operational research problems but also are closely related to uniform sampling (using Markov chains) or enumeration of solutions of a problem. Reconfiguration problems received an important attention in the last few years. For an overview of recent results on reconfiguration problems, the reader is referred to the surveys of van den Heuvel~\cite{vHeuvel13} and Nishimura~\cite{Nishimura18}.

In this paper, our reference problem is the spanning tree problem. Let $G=(V,E)$ be a connected graph on $n$ vertices. A \emph{spanning tree} of $G$ is a tree (chordless graph) with exactly $n-1$ edges. Given a tree $T$, a vertex $v$ is a \emph{leaf} if its degree is one and is an \emph{internal node} otherwise. A \emph{branching node} is a vertex of degree at least three.

In order to define valid step-by-step transformations, an adjacency relation on the set of feasible solutions is needed. Depending on the problem, there may be different natural choices of adjacency relations. 
 Let $T_1$ and $T_2$ be two spanning trees of $G$. We say that $T_1$ and $T_2$ differs by an \emph{edge flip} if there exist $e_1 \in E(T_1)$ and $e_2 \in E(T_2)$ such that $T_2 = (T_1 \setminus e_1 ) \cup e_2$. Two trees $T_1$ and $T_2$ are adjacent if one can transform $T_1$ into $T_2$ via an edge flip. A \emph{transformation} from $T_\source$ to $T_\target$ is a sequence of trees $\langle T_0 :=T_\source,T_1,\ldots,T_r:=T_\target \rangle$ such that two consecutive trees are adjacent. Ito et al.~\cite{Ito11} remarked that any spanning tree can be transformed into any other via a sequence of edge flips.  It easily follows from the exchange properties for matroid. Unfortunately, the problem becomes much harder when we add some restriction on the intermediate spanning trees.  One can then ask the following question: does it still exist a transformation when we add some constraints on the spanning tree? If not, is it possible to decide efficiently if such a transformation exists? This problem was already studied for vertex modification between Steiner trees~\cite{DBLP:conf/mfcs/MizutaHIZ19} for instance.

In this paper, we consider spanning trees with restrictions on the number of leaves. More precisely, what happens if we ask the number of leaves to be large (or small) all along the transformation? We formally consider the following problems: \medskip

\noindent
\textsc{Spanning Tree with Many Leaves} \\
\textbf{Input:} A graph $G$, an integer $k$, two trees $T_1$ and $T_2$ with at least $k$ leaves. \\
\textbf{Output:} {\sf yes} if and only if there exists a transformation from $T_1$ to $T_2$ such that all the intermediate trees have at least $k$ leaves.
\smallskip

\noindent
\textsc{Spanning Tree with At Most $k$ Leaves} \\
\textbf{Input:} A graph $G$, two trees $T_1$ and $T_2$ with at most $k$ leaves. \\
\textbf{Output:} {\sf yes} if and only if there exists a transformation from $T_1$ to $T_2$ such that all the intermediate trees have at most $k$ leaves.
\smallskip

\paragraph*{Our results.}
We prove that both variants are PSPACE-complete. In other words, we show that \STWML and \STWkL for every $k \ge 3$ are PSPACE-complete. 
This contrasts with many existing results on reconfiguration problems using edge flips which are polynomial such as matching reconfiguration~\cite{Ito11}, cycle, tree or clique reconfiguration~\cite{HanakaIMMNSSV20}. As far as we know there does not exist any PSPACE-hardness proof for any problem via edge flip. We hope that our results will help to design more.

More formally, our results are the following:

\begin{theorem}\label{thm:at-least-hardness}
 \STWML is PSPACE-complete restricted to bipartite graphs, split graphs or planar graphs.
\end{theorem}

These results are obtained from two different reductions. 
In both reductions, we need an arbitrarily large number of leaves in order to make the reduction work. In particular, one can ask the following question: is \textsc{Spanning Tree with at least $n-k$ Leaves} hard for some constant $k$ (where $n$ is the size of the instance)?

We did not solve this question but we prove that, for the ``dual'' problem, the PSPACE-hardness is obtained even for $k=3$.

\begin{theorem}
 \STWkL is PSPACE-complete for every $k \ge 3$. 
\end{theorem}

This proof is the most technically involved proof of this article and is based on a reduction from the decision problem of {\sc Vertex Cover} to the decision problem of {\sc Hamiltonian Path}. Let $(G=(V,E),k)$ be an instance of {\sc Vertex Cover}. We first show that, on the graph $H$ obtained when we apply this reduction, we can associate with any spanning tree $T$ of $H$ a vertex cover of $G$. The hard part of the proof consists of showing that {\em (i)} if $T$ has at most three leaves, then the vertex cover associated with $T$ has at most $k+1$ vertices; and {\em (ii)} each edge flip consists of a modification of at most one vertex of the associated vertex cover.

One can note that for $k=2$, the problem becomes the \textsc{Hamiltonian Path Reconfiguration} problem. We were not able to determine the complexity of this problem and we left it as an open problem. 

We complete these results by providing some polynomial time algorithms:

\begin{theorem}
\STWML can be decided in polynomial time on interval graphs, on cographs, or if the number of leaves is $n-2$.
\end{theorem}
 
 We show that \STWML can be decided in polynomial time if the number of leaves is $n-2$. As we already said, we left as an open question to determine if this result can be extended to any value $n-k$ for some fixed $k$. If such an algorithm exists, is it true that the problem is FPT parameterized by $k$?
 
 We then show that in the case of cographs, the answer is always positive as long as the number of leaves is at most $n-3$. Since there is a polynomial time algorithm to decide the problem when $k=2$ that completes the picture for cographs.

 Since the problem is known to be PSPACE-complete for split graphs by Theorem~\ref{thm:at-least-hardness} (and thus for chordal graphs), the interval graphs result is the best we can hope for in a sense.
 The interval graph result is based on a dynamic programming algorithm inspired by~\cite{BonamyB17} where it is proved that the {\sc Independent Set Reconfiguration} problem in the token sliding model is polynomial. Even if dynamic algorithms work quite well to decide combinatorial problems on interval (and even chordal) graphs, they are much harder to use in the reconfiguration setting. In particular, many reconfiguration problems become hard on chordal graphs (see e.g.~\cite{Belmonte0LMOS19,HatanakaIZ17}) since the transformations can go back and forth. 
 
 Since the problem is hard on planar graph, it would be interesting to determine its complexity on outerplanar graphs. We left this question as an open problem.

\paragraph*{Related work.}
 In the last few years, many graph reconfiguration problems have been studied through the lens of edge flips such as matchings~\cite{Ito11,BousquetHIM19}, paths or cycles~\cite{HanakaIMMNSSV20}. None of these works provide any PSPACE-hardness results, only a NP-hardness result is obtained for path reconfiguration via edge flips in~\cite{HanakaIMMNSSV20}. Even if the reachability problem is known to be polynomial in many cases, approximating the shortest transformation is often hard, see e.g.~\cite{BousquetHIM19}. Edges flips are also often considered in computational geometry, for instance to measure the distance between two triangulations. In that setting, a flip of a triangulation is the modification of a diagonal of a $C_4$ for the other one. Usually, proving the existence of a transformation is straightforward and the main questions are about the length of a transformation which is not the problem addressed in this paper.
 
 If, instead of ``edge flips'', we consider ``vertex flips'' the problems become much harder. For instance, the problem of transforming an (induced) tree into another one (of the same size) is PSPACE-complete~\cite{HanakaIMMNSSV20} (while the exchange property ensures that it is polynomial for the edge version). Mizuta et al.~\cite{DBLP:conf/mfcs/MizutaHIZ19} also showed that the existence of vertex exchanges between two Steiner trees is PSPACE-complete. But transforming subsets of vertices with some properties is known to PSPACE-complete for a long time, for instance for independent sets or cliques~\cite{Hearn:2005:PSP:1140710.1140715}.
 
 Another option would be to consider more general operations on edges. In particular, one can imagine a flip around a $C_4$ (i.e. two edges $ab$ and $cd$ are replaced by $ad$ and $bc$). This operation seems to be harder than the single edge flip since, for instance, matching reconfiguration becomes PSPACE-complete~\cite{BonamyBHIKMMW19}.
 
\paragraph*{Definitions.}
Given two sets $S_1$ and $S_2$, we denote by $S_1 \, \triangle \, S_2$ the \emph{symmetric difference} of the sets $S_1$ and $S_2$, that is $(S_1 \setminus S_2) \cup (S_2 \setminus S_1)$.

For a spanning tree $T$, every vertex of degree one is a \emph{leaf} and every vertex of degree at least two is an \emph{internal node}. A vertex of degree at least three is called a \emph{branching node}. Recall that the number of leaves of any tree $T$ is equal to $(\sum_{v \in T} (\max\{0,d_T(v)-2\}))+2$.
We denote by $in(T)$ the number of internal nodes of $T$. Note that if $T$ contains $n$ nodes, the number of leaves is indeed $n-in(T)$.

Let $G=(V,E)$ be a graph. A {\em vertex cover} $C$ of $G$ is a subset of vertices such that for every edge $e \in E$, $C$ contains at least one endpoint of $e$. $C$ is {\em minimum} if its cardinality is minimum among all vertex covers of $G$. Note that in particular, $C$ is inclusion-wise minimal and thus for every vertex $u \in C$, there is an edge $e \in E$ which is covered only by $u$. We denote by $\tau(G)$ the size of a minimum vertex cover of $G$.

Let $X,Y$ be two vertex covers of $G$. $X$ and $Y$ are \emph{{\sf TAR}-adjacent}\footnote{{\sf TAR} stands for ``Token Additional Removal''.} (resp. {\sf TJ}-adjacent) if there exists a vertex $x$ (resp. $x$ and $y$) such that $X= Y \cup \{ x \}$ or $Y= X \cup \{ x \}$ (resp. $X=Y \setminus \{ y \} \cup \{ x \}$).
We will consider the following problem: \medskip

\noindent
\textsc{Minimum {\sf TAR}-Vertex Cover Reconfiguration} \\
\textbf{Input:} A graph $G$, two minimum vertex covers $X,Y$ of size $k$. \\
\textbf{Output:} {\sf yes} if and only if there exists a sequence from $X$ to $Y$ of {\sf TAR}-adjacent vertex covers, all of size at most $k+1$.
\medskip

Similarly, one can define the {\sc Minimum {\sf TJ}-Vertex Cover Reconfiguration} ({\sc MVCR} for short) where we want to determine whether there exists a sequence of {\sf TJ}-adjacent vertex covers from $X$ to $Y$. Note that all the vertex covers must be of size $|X| = |Y| = k$.

\section{Spanning trees with few leaves}

\begin{theorem}\label{thm:stwfl-hardness2}
For every $k \ge 3$, \STWkL is PSPACE-complete. 
\end{theorem}
In order to prove Theorem~\ref{thm:stwfl-hardness2}, we will first prove it for $k=3$ in Section~\ref{sec:stwfl_3leaves} and explain how we can modify this proof in order to get the hardness for the general case in Section~\ref{sec:stwfl_moreleaves}.

\begin{theorem}\label{thm:stwfl-hardness}
 \STWThL is PSPACE-complete. 
\end{theorem}

Recall that proving an hardness result for $n-2$ internal nodes and two leaves would imply that the problem \textsc{Hamiltonian Path Reconfiguration} problem is hard, a problem left open in this paper. Even if the optimization version of the \textsc{Hamiltonian Path} problem is very hard, its reconfiguration counterpart seems ``easier'' since at each step, the modification must be around one of the two endpoints of the path. Indeed, most of the PSPACE-hardness proofs in reconfiguration follows from NCL logic (the ``classical'' problem to reduce from in reconfiguration). But in an instance of NCL logic, modifications can appear almost everywhere in the instance (under some local conditions) while, in \textsc{Hamiltonian Path Reconfiguration}, the modification has to be ``localized'' on the endpoints of the paths.

In order to prove Theorem~\ref{thm:stwfl-hardness}, we will provide a reduction from \textsc{Minimum {\sf TAR}-Vertex Cover Reconfiguration} to \STWThL. 

\begin{theorem}[Wrochna~\cite{Wrochna18}]\label{thm:stwfl_TARPSPACE}
\textsc{\TAR-Vertex Cover Reconfiguration} is PSPACE complete even for bounded bandwidth graphs.
\end{theorem}

Actually the result of Wrochna is for \textsc{Maximum Independent Set Reachability} in the Token Jumping model. However, recall that the complement of an independent set is a vertex cover. Besides, Kami\'nski et al.~\cite{KAMINSKI20129} observed that the {\sf TJ} model and {\sf TAR} model are equivalent when the threshold is the minimum value of a vertex cover plus one. Hence, the result of~\cite{Wrochna18} is equivalent to the statement of Theorem~\ref{thm:stwfl_TARPSPACE}.

The idea of the proof of Theorem~\ref{thm:stwfl-hardness} is to adapt a reduction from \textsc{Vertex cover} to \textsc{Hamiltonian Path} (for the optimization version). Let $(G=(V,E),k)$  be an instance of \textsc{Vertex Cover}. This reduction creates a graph $H(G)$ which contains a Hamiltonian path if and only if $G$ admits a vertex cover of size $k$. In particular, we will show that there is a ``canonical way'' to define a vertex cover from any Hamiltonian path. The reduction is provided in Section~\ref{sec:stwfl_optversion} together with some properties of the spanning trees with at most three leaves in $H(G)$. In order to adapt the proof in the reconfiguration setting, we need to prove that the proof is ``robust'' with respect to several meanings of the word. First, we need to show that, if we consider a spanning tree with at most three leaves in $H(G)$ then there is a ``canonical'' vertex cover of size at most $k+1$ associated with it. Proving that this vertex cover always has size at most $k-1$ is the first technical part of the proof.  
Then, for any edge flip between two spanning trees with at most three leaves, there is a corresponding ``transformation'' between the canonical vertex covers associated with them. We need to prove that for any two adjacent spanning trees in $H(G)$, their canonical vertex covers are either the same or are incident in the {\sf TAR} model (in $G$). 

Finally, we need to prove that it is possible to transform a Hamiltonian path $P_1$ (associated with a vertex cover $X$) into a Hamiltonian path $P_2$ (associated with a vertex cover $Y$) via spanning trees with at most three leaves if and only if $X$ can be transformed into $Y$ in the {\sf TAR} model. 

\subsection{The Reduction}\label{sec:stwfl_optversion}

The reduction is a classical reduction (see Theorem 3.4 of~\cite{GJ79} for a reference) from the optimization version of \textsc{Vertex Cover} to the optimization version of \textsc{Hamiltonian Path}.
Let $G$ be a graph and $k$ be an integer. We provide a reduction from \textsc{Vertex Cover} of size at most $k$ to \textsc{Hamiltonian path}. Let us construct a graph $H(G)$ (abbreviated into $H$ when no confusion is possible) as follows:

\begin{figure}[bt]
    \centering
    \includegraphics[scale=.75]{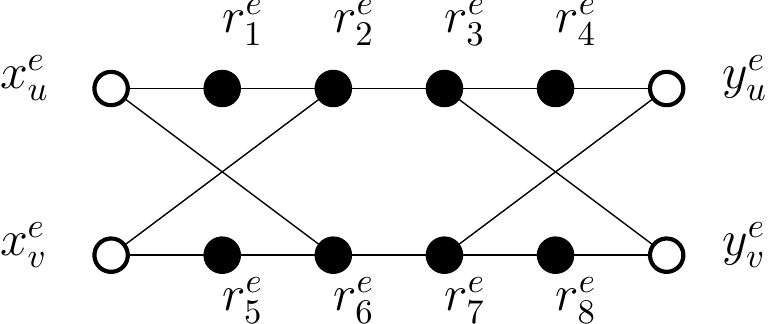}
    \caption{Edge-gadget corresponding to the edge $e=uv$. The white vertices are the only ones connected to the outside.}
    \label{fig:stwl_reduction1}
\end{figure}

\paragraph*{Construction of $H(G)$.}
For each edge $e=uv$ of $G$, we create the following \emph{edge-gadget} $\mathcal{G}_e$ represented in Figure~\ref{fig:stwl_reduction1}. The edge-gadget $\mathcal{G}_e$ has four \emph{special vertices} denoted by $x_u^e,x_v^e,y_u^e,y_v^e$. The vertices $x_u^e$ and $x_v^e$ are called the \emph{entering vertices} and $y_u^e$ and $y_v^e$ the \emph{exit vertices}. The gadget contains $8$ additional vertices denoted by $r_1^e,\ldots,r_8^e$. When $e$ is clear from context, we will omit the superscript. The graph induced by these twelve vertices is represented in Figure~\ref{fig:stwl_reduction1}. The vertices $r_1^e,\ldots,r_8^e$ are \emph{local vertices} and their neighborhood will be included in the gadget. The only vertices connected to the rest of the graphs are the special vertices.

We add an independent set $Z:=\{ z_1,\ldots,z_{k+1} \}$ of $k+1$ new vertices to $V(H)$. And we finally add to $V(H)$ two more vertices $s_1,s_2$ in such a way that $z_1$ (resp. $z_{k+1}$) is the only neighbor of $s_1$ (resp. $s_2$) in $H(G)$. 

Since $s_1$ and $s_2$ have degree one in $H(G)$, $s_1$ and $s_2$ are leaves in any spanning tree of $H(G)$. In particular, the two endpoints of any Hamiltonian path of $H(G)$ are necessarily $s_1$ and $s_2$.  

Let us now complete the description of $H(G)$ by explaining how the special vertices are connected to the other vertices of $H(G)$. Let $u \in V(G)$. Let $E'=e_1,\ldots,e_\ell$ be the set of edges incident to $u$ in an arbitrary order. We connect $x_u^{e_1}$ and $y_u^{e_\ell}$ to all the vertices of $Z$. For every $1\le i \le \ell-1$, we connect $y_u^{e_i}$ to $x_u^{e_{i+1}}$. The edges $y_u^{e_i}x_u^{e_{i+1}}$ are called the \emph{special edges of $u$}. The \emph{special edges} of $H(G)$ are the union of the special edges for every $u \in V(G)$ plus the edges incident to $Z$ but $s_1z_1$ and $s_2z_{k+1}$. 

This completes the construction of $H(G)$ (see Figure~\ref{fig:reduction-hamiltonian-VC} for an example).

\begin{figure}[bt]
    \centering
    \includegraphics[width=0.98\textwidth]{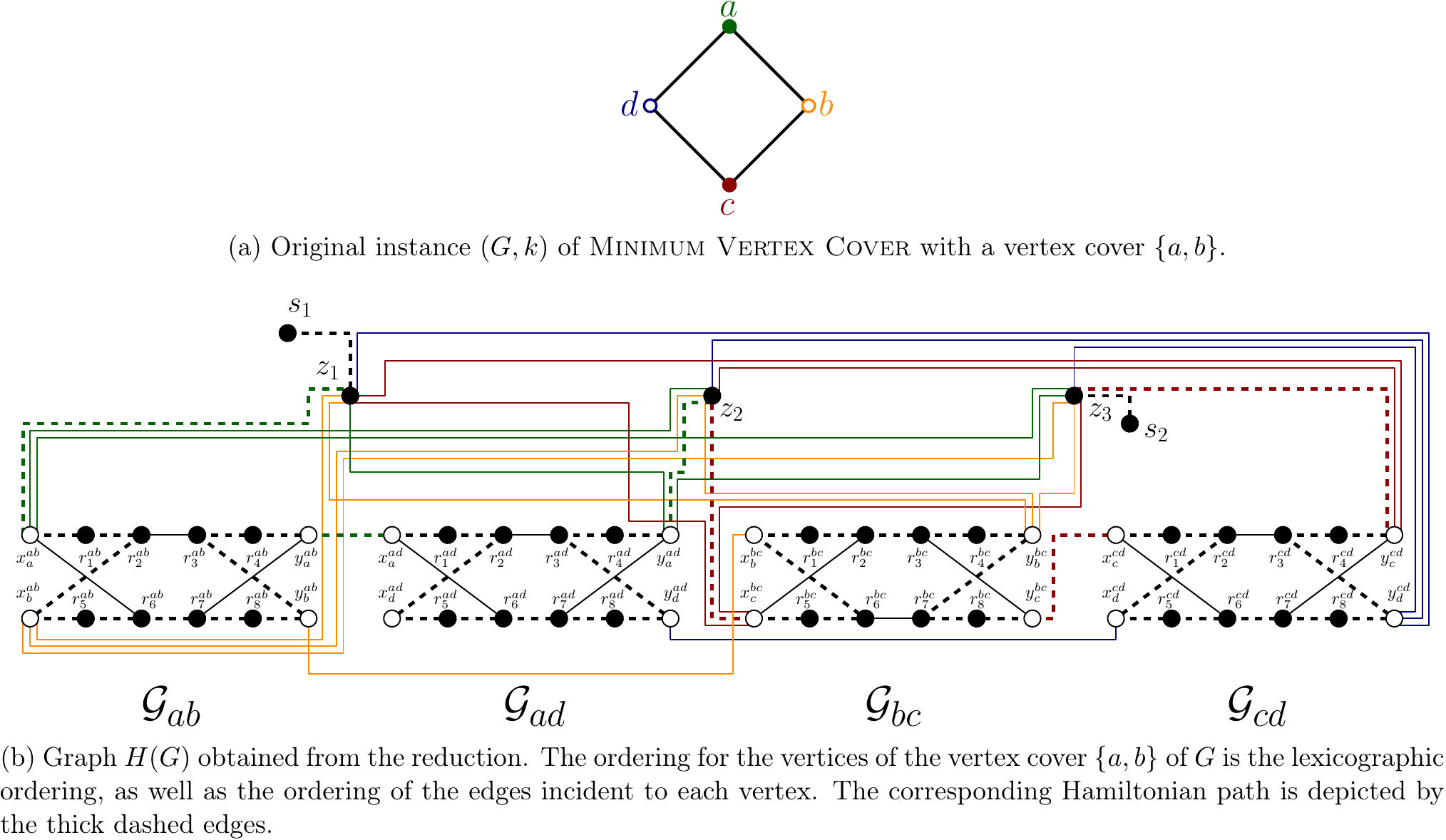}
    \caption{Illustration of the reduction of Theorem~\ref{thm:stwfl-hardness}.}
    \label{fig:reduction-hamiltonian-VC}
\end{figure}

\subsection{Basic properties of \texorpdfstring{$H(G)$}{H(G)}} 
\begin{remark}\label{rem:stwfl_nots1s2}
If $T$ is a spanning tree of $H(G)$ with at most $\ell$ leaves, then at most $\ell-2$ of them are in $V(H) \setminus \{ s_1,s_2 \}$.
\end{remark}

\paragraph*{Definitions and notations.} 
For a spanning tree $T$, we say that an edge-gadget \emph{contains a leaf} if one of the twelve vertices of the edge-gadget is a leaf of $T$. If the spanning tree is a Hamiltonian path, Remark~\ref{rem:stwfl_nots1s2} ensures that no edge-gadget contains a leaf. Besides, at most one edge-gadget contains a leaf if $T$ is a spanning tree with at most three leaves.
An edge-gadget \emph{contains a branching node} of $T$ if one of the twelve vertices of the gadget is a vertex of degree at least three. Any spanning tree with at most three leaves indeed contains at most one branching node. 

Let $T$ be a spanning tree of $H(G)$. An edge-gadget is \emph{irregular} if at least one of its twelve vertices is not of degree two in $T$, i.e. if it contains a branching node or a leaf. An edge-gadget is \emph{regular} if it is not irregular. By abuse of notation we say that $e \in E(G)$ is regular (resp. irregular) if the edge-gadget of $e$ is regular (resp. irregular). A vertex $u$ is \emph{regular} if every edge incident to $u$ is regular. The vertex $u$ is \emph{irregular} otherwise.

Let $S$ be a subset of vertices of $H(G)$. We denote by $\delta_T(S)$ the set of edges with exactly one endpoint in $S$. When there is no ambiguity, we omit the subscript $T$. Moreover, if $S$ is the singleton $\{u\}$, we write $\delta_T(u)$ for $\delta_T(\{u\})$.
Given an edge $e$ of $G$ and a spanning tree $T$ of of $H(G)$, $\delta_T(e)$ denotes the set of edges of $T$ with exactly one endpoint in the edge-gadget $\mathcal{G}_e$ of $e$. 
The restriction $T(\mathcal{G}_e)$ of a spanning tree $T$ \emph{around an edge-gadget} $\mathcal{G}_e$ is the set of edges with both endpoints in $\mathcal{G}_e$ plus the edges of $\delta_T(\mathcal{G}_e)$ (which are considered as ``semi edge'' with one endpoint in $\mathcal{G}_e$).

\begin{lemma}\label{lem:stwml_intersectionedgegadget}
Let $T$ be a spanning tree of $H$ and $\mathcal{G}$ be a regular edge-gadget. Then the tree $T$ around the edge-gadget $\mathcal{G}$ is one of the two graphs represented in Figure~\ref{fig:stwl_reduction2}. Note that the graph of Figure~\ref{fig:stwl_reduction2}(b) has to be considered up to symmetry between $u$ and $v$.
\end{lemma}

\begin{figure}[bt]
    \centering
    \includegraphics[scale=.75]{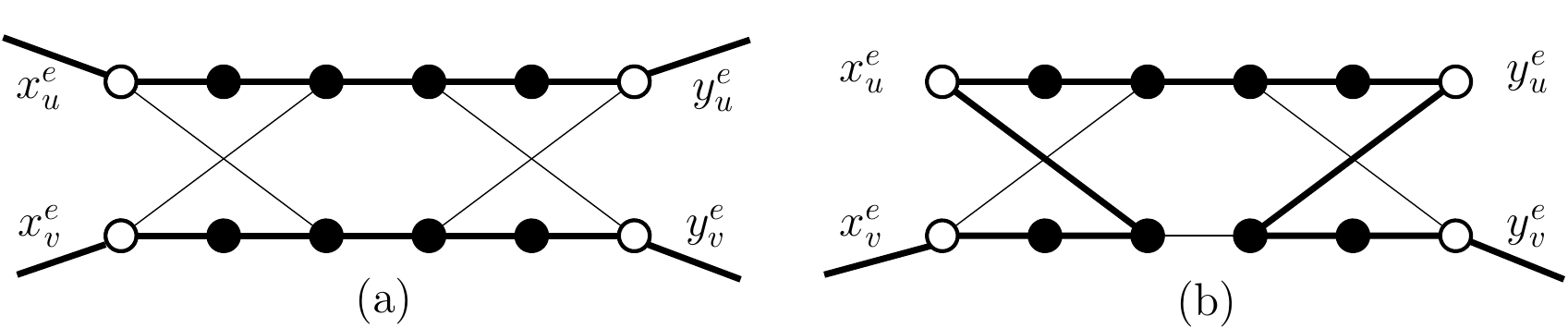}
    \caption{The two possible sub-graphs around a regular edge-gadget $\mathcal{G}$. Bold edges are edges in the tree. Edges with one endpoint in the gadget are edges of $\delta(\mathcal{G})$.}
    \label{fig:stwl_reduction2}
\end{figure}

In order to prove Lemma~\ref{lem:stwml_intersectionedgegadget}, we will need the following lemma that will be useful all along our proof:

\begin{lemma}\label{clm:stwml_ham}
Let $R$ be the graph restricted to an edge-gadget. There is no Hamiltonian path from one vertex of $\{ x_u^e,y_u^e \}$ to one vertex of $\{x_v^e,y_v^e \}$ in $R$.
\end{lemma}

\begin{proof}
By contradiction. Let us denote by $w_1,w_2$ the two endpoints of a Hamiltonian path $P$.
If $w_1,w_2$ are the two entering (resp. exit) vertices, then both exit (resp. entering) vertices must have degree two in $P$. If both exit vertices have degree two, then one of $r_3r_4$ or $r_7r_8$ do not exist in $P$ since otherwise $P$ admits a cycle. And then $r_4$ or $r_8$ are leaves of $P$, a contradiction since $P$ is a Hamiltonian path in $R$. Similarly, the same holds if both entering vertices have degree two.

So, by symmetry, we can assume that $w_1=x_u^e$ and $w_2 = y_v^e$. Since $x_v^e$ and $y_u^e$ have degree two and all the local vertices have degree two in $P$, the subpaths $x_u^er_1r_2x_v^er_5r_6$ and $r_3r_4y_u^er_7r_8y_v^e$ are in $P$. It is impossible to connected these two paths into a Hamiltonian path in $R$, a contradiction.
\end{proof}

Let us now prove Lemma~\ref{lem:stwml_intersectionedgegadget}:

\begin{proof}
Remark that since all the vertices of the edge-gadget $\mathcal{G}_e$ have degree two in $T$, the number of edges with one endpoint in the gadget is even (the subgraph of $T$ induced by the vertices of $\mathcal{G}_e$ being a union of paths). Moreover, since $r_1,r_4,r_5,r_8$ are not leaves of $T$ and have degree two in $H(G)$, both edges incident to them are in $T$. So the number of edges of $\delta_T(\mathcal{G}_e)$ incident to each of $x_u^e,y_u^e,x_v^e,y_v^e$ is either zero or one. In particular, $|\delta_T(\mathcal{G}_e)| \le 4$.

If $|\delta_T(\mathcal{G}_e)|=2$, then, since the edge-gadget is regular, the restriction of $T$ to the edge-gadget is a Hamiltonian path $P$. By Lemma~\ref{clm:stwml_ham}, the endpoints of $P$ cannot be one vertex of $\{ x_u,y_u \}$ and one vertex of $\{x_v, y_v \}$. So, by symmetry, we can assume that the endpoints of $P$ are $x_u$ are $y_u$. Since, $r_1,x_v,r_5,r_8,y_v,r_4$ have degree two in the subgraph induced by the edge-gadget, it forces all the edges of the gadget but $x_ur_6, y_ur_7, r_6r_7$ and $r_2r_3$ to be in $P$. Since $P$ is an Hamiltonian path from $x_u$ to $y_u$, $r_5r_6 \in E(T)$ which gives the graph of Figure~\ref{fig:stwl_reduction2}(b) (up to symmetry.).  

So we can now assume that $|\delta_T(\mathcal{G}_e)|=4$. Since at most one edge of $\delta_T(\mathcal{G}_e)$ is incident to each special vertex, all these vertices have degree one in the subtree induced by the vertices of $\mathcal{G}_e$. So, the subforest induced on the gadget must be a union of two paths. Since $r_1,r_4,r_5$ and $r_8$ have degree two, the only way to complete this set of edges into a Hamiltonian path provides the graph of Figure~\ref{fig:stwl_reduction2}(a), which completes the proof.
\end{proof}

If $P$ is a Hamiltonian path of $H$, then Remark~\ref{rem:stwfl_nots1s2} ensures that all its edge-gadgets are regular. And then, by Lemma~\ref{lem:stwml_intersectionedgegadget}, for every edge-gadget $\mathcal{G}$, the graph around $\mathcal{G}$ is one of the two graphs of Figure~\ref{fig:stwl_reduction2}. 

\paragraph*{Vertex Cover and Hamiltonian Path.}
Let us assume that $G$ has a vertex cover $X=\{ v_1,\ldots,v_k \}$ of size $k$. We claim that the following set of edges $F$ induces a Hamiltonian path in $H(G)$. We start with $F= \emptyset$. For every $i \le k$, we add to $F$ the edge between  $z_i$ and the entering vertex of the first edge of $v_i$ and the edge between $z_{i+1}$ an the exit vertex of the last edge of $v_i$. For every $v_i \in X$, all the special edges of $v_i$ are added to $F$. The edges $s_1z_1$ and $s_2z_{k+1}$ are also in $F$. We claim that, for each edge-gadget $\mathcal{G}$ corresponding to the edge $uv$, either two edges or four edges of $F$ have exactly one endpoint in $F$. Indeed, if none of them are selected, then by construction of $F$, neither $u$ nor $v$ are in $X$, a contradiction since $X$ is a vertex cover of $G$. Moreover, by construction of $F$, $x_v^e$ is an endpoint of an edge of $F$ if and only if $y_v^e$ also is. Note moreover that: {\em (i)} no local vertex of the edge-gadget is incident to an edge of $F$, {\em (ii)} special vertices are incident to at most one, and {\em (iii)} vertices of $Z$ are incident to two of them. 
So in order to complete $F$ into a Hamiltonian path, we add the edges of Figure~\ref{fig:stwl_reduction2}(a) or (b) depending if two or four edges of the current set $F$ are incident to a vertex of the edge-gadget (two when one endpoint is in $X$, four is both of them are in $X$).
The set $F$ induces a Hamiltonian path, as proved in~\cite{GJ79}. This Hamiltonian path is called \emph{a Hamiltonian path associated with the vertex cover $X$}\footnote{Note that there might be several Hamiltonian paths associated with the same vertex cover since the the path depends on the ``ordering'' of $X$. Indeed we have to choose which entering vertex is attached to $z_1,z_2,\ldots,z_k$ which gives a natural ordering of $X$.}

Conversely, let us explain why we can associate with every Hamiltonian path $P$ a vertex cover. 
Let $\mathcal{G}$ be an edge-gadget. The graph $H[\mathcal{G}]$ is the subgraph induced by the twelve vertices of the edge-gadget. (Note that the subgraph of $P$  induced by $\mathcal{G}$ is not the graph around $\mathcal{G}$, which contains the semi-edges leaving $\mathcal{G}$.

\begin{lemma}\label{lem:stwml_normal}
Let $G$ be a graph, $T$ be a spanning tree of $H(G)$, and $u$ be a regular vertex of $T$. If there exists an edge $e \in E(G)$ with endpoint $u$ such that $x_u^e$ or $y_u^e$ has degree one in the subgraph of $T$ induced by the vertices of $H[\mathcal{G}_e]$, then, for every edge $e'$ with endpoint $u$, $x_u^{e'}$ and $y_u^{e'}$ have degree one in the subgraph of $T$ induced by the vertices of $H[\mathcal{G}_{e'}]$.  \\
In particular, there is an edge of $T$ between $Z$ and the first entering vertex of $u$ and an edge between $Z$ and the last exit vertex of $u$.
\end{lemma}

\begin{proof}
By symmetry, $x_u^e$  has degree one in the subgraph of $T$ induced by the vertices of. Since the graph around the gadget is one of the two graphs  of $H[\mathcal{G}_e]$. 
In Figure~\ref{fig:stwl_reduction2} (which corresponds to the only possible restrictions of $T$ around a regular edge-gadget), for both $x_u^e$ and $y_u^e$, an edge of $T$ is leaving the gadget. If $e$ is the first (resp. last) edge of $u$, then there is an an edge linking $x_u^e$ (resp. $y_u^e$) to $Z$. Otherwise, let us denote by $e'$ (resp. $e''$) the edge before (resp. after) $e$ in the order of $u$. The only edge incident to $x_u^e$ (resp. $y_u^e$) in $\delta_T(\mathcal{G}_{e'})$ is $x_u^ey_u^{e'}$ (resp. $y_u^ex_u^{e''}$). Since $u$ is regular, both $x_u^ey_u^{e'}$ and $y_u^ex_u^{e''}$ are in $T$. And then we can repeat the same argument on $e'$ (resp. $e''$) until we reach the first (resp. last) edge of $u$.
\end{proof}

If, for a regular vertex $u$ and an edge $e=uv$, $x_u^e$ or $y_u^e$ have degree one in $H[\mathcal{G}_e]$, then there is a path between two vertices of $Z$ passing through all the special vertices $x_u^{e'}$ and $y_u^{e'}$ for every $e'$ incident to $u$ and all the vertices on this path have degree two. Note that the union of all such vertices forms a vertex cover of $G$.

\subsection{Reconfiguration hardness}

\subsubsection{Defining a vertex cover}\label{sec:stwml_S}
Let $T$ be a spanning tree with at most three leaves. By Lemma~\ref{lem:stwml_intersectionedgegadget}, for every edge-gadget $\mathcal{G}_e$, if $T(\mathcal{G}_e)$ is not one of the two graphs of Figure~\ref{fig:stwl_reduction2}, $\mathcal{G}_e$ contains a branching node or a leaf. So Remark~\ref{rem:stwfl_nots1s2} implies:

\begin{remark}\label{rem:stwml_le2}
There are at most two irregular edge-gadgets. Thus there are at most four irregular vertices.
\end{remark}

Indeed, if $T$ has two leaves, all the edge-gadgets are regular. If $T$ has three leaves, the third leaf must be in an edge-gadget, creating an irregular edge-gadget. And this leaf might create a new branching node which might be in another edge-gadget than the one of the third leaf. So the number of irregular edge-gadget is at most two, and thus the number of irregular vertices is at most four (if the edges corresponding to these two edge-gadgets have pairwise distinct endpoints).

Let $T$ be a spanning tree of $H(G)$ with at most three leaves. 
A vertex $v$ is \emph{good} if there exists an edge $e=vw$ for $w \in V(G)$ such that $x_v^e$ or $y_v^e$ has degree one in the subtree of $T$ induced by the twelve vertices of the edge-gadget of $e$. In other words, if we simply look at the edges of $T$ with both endpoints in $\mathcal{G}_e$, $x_v^e$ or $y_v^e$ has degree one (or said again differently, $x_v^e$ or $y_v^e$ are adjacent to exactly one local vertex).
Let us denote by $S(T)$ the set of good vertices. 

\begin{lemma}\label{lem:stwml_S_VC}
Let $T$ be a spanning tree with at most three leaves of $H(G)$ and $e=uv$ be an edge of $G$. At least one special vertex of the edge-gadget $\mathcal{G}_e$ has degree one in the subgraph of $T$ induced by the vertices of $\mathcal{G}_e$. \\
In particular, $S(T)$ is a vertex cover. 
\end{lemma}

\begin{proof}
Let $R$ be the subgraph of $H(G)$ induced by the vertices of $\mathcal{G}_e$. Let $T'$ be the restriction of $T$ to $R$. Assume by contradiction that none of the four special vertices have degree one in $T'$. Since special vertices $Y$ have degree two in $R$, the special vertices have degree zero or degree two in $T'$. 

We claim that the number of special vertices of degree zero is at most one. Indeed, if $x_u^e$ (resp. $y_u^e,x_v^e,y_v^e$) has degree zero in $T'$, then $r_1$ (resp. $r_4,r_5,r_8$) is a leaf of $T$. Since $T$ has at most three leaves, Remark~\ref{rem:stwfl_nots1s2} ensures that at most one of them have degree one in $T'$ and thus at least three vertices of $Y$ have degree two in $T'$.

So, we can assume without loss of generality that both entering vertices have degree two in $T'$.
Then, $x_ur_1,x_ur_6,x_vr_2$ and $x_vr_5$ are edges. Since $T$ is a tree, one of $r_1$ or $r_5$ are leaves. 
Now if $y_u$ (resp. $y_v$) has degree zero in $T'$ then $r_4$ (resp. $r_8$) is a leaf of $T$. And, if both $y_u,y_v$ have degree two, then $r_4$ or $r_8$ are leaves. In both cases, we have a contradiction with Remark~\ref{rem:stwfl_nots1s2}.
\end{proof}

So, for every tree $T$ with at most three leaves, $S(T)$ is a vertex cover. We say that $S(T)$ is \emph{the vertex cover associated with $T$}.

\subsubsection{ST-reconfiguration to VCR}\label{sec:stwfl_3leaves}

The goal of this section is to prove that an edge flip reconfiguration sequence between spanning trees with at most three leaves in $H(G)$ provides a {\sf TAR} vertex cover reconfiguration sequence in $G$. So we want to prove that {\em (i)} for every spanning tree $T$ with at most three leaves, $|S(T)| \le k+1$. And {\em (ii)}, for every tree $T'$ obtained via an edge flip from $T$, $|S(T) \setminus S(T')|+|S(T') \setminus S(T) | \le 1$.

\begin{lemma}\label{lem:stwml_1edgetoZ}
Let $T$  be a spanning tree of $H(G)$ with at most three leaves. Let $u$ be a vertex of $G$ and $e$ be an irregular edge with endpoint $u$. Assume moreover that no edge before $u$ (resp. after $u$) in the ordering of $u$ are irregular. Then if there is an edge of $\delta_T(\mathcal{G}_e)$ incident to $x_u^e$ (resp. $y_u^e$) then there is an edge between $Z$ and the first (resp. last) entering (resp. exit) vertex of $u$. 
\end{lemma}

\begin{proof}
Assume that an edge of $\delta_T(\mathcal{G}_e)$ is incident to $x_u^e$. Since $\mathcal{G}_e$ is the unique irregular edge-gadget for $u$, we can conclude using the arguments of Lemma~\ref{lem:stwml_normal}.
\end{proof}

Let us now prove that $|S(T)| \le k+1$ for any spanning tree $T$ with at most three leaves. When no confusion is possible, we will write $S$ for $S(T)$.

\begin{lemma}\label{lem:stwml_k+1}
Every spanning tree $T$ of $H(G)$ with at most three leaves satisfies $|S(T)| \le k+1$.
\end{lemma}

\begin{proof}
Assume by contradiction that $|S| \ge k+2$. By Remark~\ref{rem:stwml_le2}, at least $k-2$ vertices of $S$ are regular. By Lemma~\ref{lem:stwml_normal}, for each regular vertex $w \in S$, there is an edge of $T$ between $Z$ and the first entering vertex of $w$ and $Z$ and the last exit vertex of $w$. So at least $2k-4$ edges of $\delta_T(Z)$ are incident to regular vertices. Moreover two edges of $\delta_T(Z)$ are incident to $s_1$ and $s_2$. So, $T$ already has $2k-2$ edges in $\delta_T(Z)$. Since $|Z|=k+1$ and $T$ has at most three leaves, Remark~\ref{rem:stwfl_nots1s2} ensures that $\delta_T(Z)$ has size $2k+1, 2k+2$ or $2k+3$. Indeed, if either all the vertices of $Z$ have degree two or if $Z$ contains both the vertex of degree three and the vertex of degree one, then $|\delta_T(Z)|=2k+2$. Otherwise, if $Z$ only contains the vertex of degree one (resp. three), and then $|\delta_T(Z)|=2k+1$ (resp. $2k+3$). Moreover, if there is no irregular edge-gadget then, since $|S| \ge k+2$, Lemma~\ref{lem:stwml_normal} ensures that $Z$ is incident to at least $2k+4$ edges, a contradiction. So there is one or two irregular edge-gadgets by Remark~\ref{rem:stwml_le2}.
\medskip

\noindent \textbf{Case 1.} $T$ has exactly one irregular edge-gadget $\mathcal{G}_e$ for $e=uv$. \\
Since $|S| \ge k+2$, $k$ vertices are regular (otherwise the number of edges incident to $Z$ would be at least $2k+4$ using the argument above, a contradiction). So by Lemma~\ref{lem:stwml_normal}, $2k$ edges of $\delta_T(Z)$ are incident to regular vertices and two are incident to $s_1$ and $s_2$. So it already gives $2k+2$ edges in $\delta_T(Z)$. Moreover, since $T$ is connected, at least one edge is in $\delta_T(\mathcal{G}_e)$. So by Lemma~\ref{lem:stwml_1edgetoZ}, exactly one edge of $T$ is in $\delta_T(\mathcal{G}_e)$. Note that it already gives $2k+3$ edges incident to $Z$ so a vertex of $Z$ has degree three. And then, in $T$, all the vertices of $\mathcal{G}_e$ but at most one have degree two and the last one have degree one. Moreover, $|\delta_T(\mathcal{G}_e)|=1$.

Let $R$ be the graph restricted to $\mathcal{G}_e$ and $T'$ be the subforest of $T$ restricted to $R$. Since both $u$ and $v$ are in $S$, at least one vertex $v_1$ in $\{ x_u^e,y_u^e \}$ (resp. $v_2$ in $\{x_v^e,y_v^e \}$) has degree one in $R$. Since all the vertices have degree two in $T$ but at most one and $|\delta_T(R)|=1$, the graph $T'$ on $V(\mathcal{G}_e)$ is a Hamiltonian path between $v_1$ and $v_2$. In particular, all the local vertices must have degree two in $T'$.
By Lemma~\ref{clm:stwml_ham}, there is no Hamiltonian path between $v_1$ an $v_2$, a contradiction.
\medskip

\noindent \textbf{Case 2.} There are two irregular edge-gadgets $\mathcal{G}_1$ and $\mathcal{G}_2$.

Since each special edge-gadget of $T$ contains a vertex of degree one or a vertex of degree three by Lemma~\ref{lem:stwml_intersectionedgegadget}, all the vertices of $Z$ have degree two in $T$. So, $|\delta_T(Z)|=2k+2$. Since we have seen that at least $2k-4$ edges of $\delta_T(Z)$ are incident to regular vertices, there are at most four edges between $Z$ and special vertices of irregular vertices.
\smallskip

\noindent \textit{Case 2.a.} The two irregular edge-gadgets are not endpoint disjoint. 

We denote by $u_1u_2$ and $u_2u_3$ the two edges of the irregular edge-gadgets.  We can assume without loss of generality that the edge-gadget of $u_1u_2$ contains a vertex of degree one and the one of $u_2u_3$ contains a vertex of degree three. 

Since $u_1u_2$ (resp. $u_2u_3$) is the unique irregular edge incident to $u_1$ (resp. $u_3$),  all the edges incident to $u_1$ (resp. $u_3$) before and after $u_1u_2$ (resp $u_2u_3$) in the ordering of $u_1$ (resp. $u_3$) are regular. So if there is an edge of $\delta(\mathcal{G}_{u_1u_2})$ (resp. $\delta(\mathcal{G}_{u_2u_3})$) incident to the entering or exit vertex of $u_1$ (resp. $u_3$), Lemma~\ref{lem:stwml_1edgetoZ} ensures that this edges creates an additional edge incident to $Z$.

Let $a\ge 0$ such that $|S| = k+2+a$. Let us first prove that $a=0$. Since there are three irregular vertices, there are at least $k-1+a$ regular vertices. So by Lemma~\ref{lem:stwml_normal}, at least $2k-2+2a$ edges of $\delta_T(Z)$ are incident to regular vertices and two are incident to $s_1$ and $s_2$ by Remark~\ref{rem:stwfl_nots1s2}.
So in total, it already gives $2k+2a$ edges incident to $Z$. Since $|\delta_T(Z)| = 2k+2$, if $a>0$ then there is no edge between $Z$ and an entering or exit vertex of an irregular vertex. 

So no edge of $\delta(\mathcal{G}_{u_1u_2})$ is incident to the entering or exit vertex of $u_1$ and the same holds for $u_3$ in $\delta(\mathcal{G}_{u_2u_3})$ by Lemma~\ref{lem:stwml_1edgetoZ} (since $u_1u_2$ are and $u_2u_3$ are the only irregular edges incident to respectively $u_1$ and $u_3$). Up to symmetry, we can assume that $u_1u_2$ is before $u_2u_3$ in the ordering of $u_2$. So Lemma~\ref{lem:stwml_1edgetoZ} ensures no edge is not incident to the entering vertex of $u_2$ in $\delta(\mathcal{G}_{u_1u_2})$ and the exit vertex of $u_2$ in $\delta(\mathcal{G}_{u_2u_3})$ (these edges are the only irregular edge-gadgets containing $u_2$). So  if $\delta_T(\mathcal{G}_{u_1u_2})$ (resp. $\delta_T(\mathcal{G}_{u_2u_3})$ is not empty, it can only contain an edge incident to $y_{u_2}^{u_1u_2}$ (resp. $x_{u_2}^{u_2u_3}$).

But since $T$ is connected, at least one edge has to leave from $\mathcal{G}_{u_1u_2}$ and $\mathcal{G}_{u_2u_3}$. So $T$ have to contain the edges leaving $y_{u_2}^{u_1u_2}$ and $x_{u_2}^{u_2u_3}$\footnote{Note that it might be the same edge if $u_1u_2$ and $u_2u_3$ are consecutive in the ordering of $u_2$ .}. But since the gadgets between them are regular, all the vertices between $y_{u_2}^{u_1u_2}$ and $x_{u_2}^{u_2u_3}$ in $T$ have degree two and does not contain any vertex of $Z$. And then the vertices of the two edge-gadgets cannot be in the connected component of $s_1$, a contradiction.

So we must have $|S|=k+2$ and $u_1,u_2$ and $u_3$ are in $S$. Indeed, there are $k-1$ regular vertices in $S$ and at most three irregular vertices candidates to be in $S$.

Let $e_1=u_1u_2$. Let $R$ be the graph restricted to $\mathcal{G}_{u_1u_2}$ and $T'$ be the subforest of $T$ restricted to $R$.
Since $\mathcal{G}_{e_1}$ does not contain any vertex of degree three and contains exactly one leaf, $T'$ is a union of paths (some of them might be reduced to a single vertex). Moreover, since $T$ has at most one leaf distinct from $s_1,s_2$, at most one local vertex (whose neighborhood is completely included in the edge-gadget) is a leaf of a path in $T'$. Since $T'$ contains a leaf and no vertex of degree at least three, $|\delta(\mathcal{G}_{u_1u_2})|$ is odd (since the sum of the degrees of $V(\mathcal{G}_{u_1u_2})$ is even in $T'$ and odd in $T$ and the difference only consists of edges in $\delta(\mathcal{G}_{u_1u_2})$).
If an entering or exit vertex contributes for two edges in $\delta(\mathcal{G}_{u_1u_2})$, one of its local neighbors is a leaf (since this vertex has degree at most two by assumption and one of its local neighbors has degree exactly two in $H$). So at most one edge incident to each -but at most one- entering and exit vertices is in $\delta(\mathcal{G}_{u_1u_2})$. Thus we have $|\delta(\mathcal{G}_{u_1u_2})| \in \{ 1, 3, 5\}$. 

First assume $|\delta(\mathcal{G}_{u_1u_2})|=5$, then there are two edges of $\delta(\mathcal{G}_{u_1u_2})$ incident to the same special vertex of the gadget. By construction of $H(G)$, a special vertex of $\mathcal{G}_{u_1u_2}$ is either incident to exactly one edge of $\delta(\mathcal{G}_{u_1u_2})$ if it is not the first entering or last exit vertex, or all the edges of $\delta(\mathcal{G}_{u_1u_2})$ incident to it goes to $Z$. So two edges of $\delta(\mathcal{G}_{u_1u_2})$ are between $Z$ and a special vertex of $\mathcal{G}_{u_1u_2}$.
So it already creates two new edges incident to $Z$. Moreover, since $|\delta(\mathcal{G}_{u_1u_2})|=5$, at least one edge leaving the gadget is incident to each entering or exit vertex. So by Lemma~\ref{lem:stwml_1edgetoZ}, since $u_1u_2$ is the only irregular gadget for $u_1$, it creates at least one more edge in $\delta_T(Z)$. Since $\delta_T(Z)$ already contains $2k-2$ edges incident to entering or exit vertices of the $k-1$ regular vertices, and two edges incident to $s_1$ and $s_2$, we have $|\delta_T(Z)| \ge 2k+3$, a contradiction. So from now on, we can assume that $|\delta(\mathcal{G}_{u_1u_2})| \in  \{ 1,3\}$.

Since $u_1 \in S$, an entering or exit vertex of $u_1$ has degree one in the restriction of $T$ to some edge-gadget containing $u_1$. If an entering or exit vertex of $u_1$ has degree one in the subtree $T'$ of $T$ restricted to the edge-gadget for an edge distinct from $u_1u_2$, then Lemma~\ref{lem:stwml_1edgetoZ} ensures that there is an edge between $Z$ and the first entering vertex of the last exit vertex of $u_1$. Now assume that at least one vertex of $x_{u_1}^{u_1u_2},y_{u_1}^{u_1u_2}$ have degree one in $T'$. Either an edge of $T$ incident to $x_{u_1}^{u_1u_2}$ or $y_{u_1}^{u_1u_2}$ leaves the edge-gadget, and then one edge goes to $Z$ by Lemma~\ref{lem:stwml_1edgetoZ}. Otherwise, w.l.o.g., $x_{u_1}$ has degree one in $T'$ and in $T$. So all the other vertices of the edge-gadget have degree two in $T$. So free to virtually add an edge between $x_{u_1}$ and the rest of the graph, the gadget becomes regular and then by Lemma~\ref{lem:stwml_intersectionedgegadget}, the vertex $y_{u_1}$ has an edge to the rest of the graph (in $T$), which finally goes to $Z$ by Lemma~\ref{lem:stwml_1edgetoZ}. So, there is at least one of $\delta_T(Z)$ incident to a special vertex of $u_1$.

Recall that $\mathcal{G}_{u_2u_3}$ contains a vertex of degree three and no leaves. Let us prove that because of this edge-gadget, we can add two edges incident to $Z$. If two of the three edges of the degree three vertex are in  $\delta(\mathcal{G}_{u_2u_3})$, we have already seen that, by definition of $H(G)$, the other endpoints of these edges are in $Z$. And then the conclusion follows. The restriction $T''$ of $T$ to the vertices of $\mathcal{G}_{u_2u_3}$ is a forest. Note that the leaves of $T''$ can only be special vertices since all the vertices of $\mathcal{G}_{u_2u_3}$ have degree at least two in $T$. If $T''$ has at least three leaves, then by Lemma~\ref{lem:stwml_1edgetoZ}, at least two of them creates an edge incident to $Z$ since the only one which does not create it is $x_{u_2}^{u_2u_3}$. Indeed, by Lemma~\ref{lem:stwml_1edgetoZ}, all the edges of $\delta(\mathcal{G}_{u_2u_3})$ incident to a special vertex of $u_3$ immediately creates an edge incident to $Z$. The same holds for $y_{u_2}^{u_2u_3}$ since $u_2u_3$ is the last irregular edge incident to $u_3$. So if $T''$ has three leaves, it creates two edges incident to $Z$ (indeed three edges are leaving the edge-gadget and only the one, if it exists, incident to $x_{u_2}^{u_2u_3}$ does not create an edge incident to $Z$).
So we can assume that $T''$ has exactly two leaves and then the degree three vertex is an entering or exit vertex. Since this vertex has degree two in $T''$, $T''$ contains two other leaves. And again there are three distinct special vertices incident to an edge of $\delta_T(\mathcal{G}_{u_2u_3})$. And as in the previous case, Lemma~\ref{lem:stwml_1edgetoZ} ensures that at least two of them are creating one new edge incident to $Z$. So in both cases, the number of edges of  $\delta_T(Z)$ incident to entering or exit vertices of $u_2,u_3$ is at least two.

So $|\delta_T(Z)| \ge 2k+3$, a contradiction.
\medskip

\noindent \textit{Case 2.b.} The two irregular edge-gadgets are endpoint disjoint.

\noindent
Let $u_1u_2$ and $u_3u_4$ be the two irregular edges. Let $\mathcal{G}_1:=\mathcal{G}_{u_1u_2}$ and and $\mathcal{G}_2:=\mathcal{G}_{u_3u_4}$. 
Note that since $u_1u_2$ and $u_3u_4$ are the unique irregular edges for respectively $u_1,u_2,u_3,u_4$, all the edges leaving these edge-gadgets create an edge incident to $Z$ by Lemma~\ref{lem:stwml_1edgetoZ}.
Since there are at most four edges between $Z$ and special vertices of irregular vertices, we have $|\delta_T(\mathcal{G}_1)|+|\delta_T(\mathcal{G}_2)| \le 4$.
Let us prove by contradiction that $|\delta_T(\mathcal{G}_1)|+|\delta_T(\mathcal{G}_2)| > 4$.

Let us first prove that the number of regular vertices is exactly $k-2$. We have already seen that it has to be at least $k-2$. Assume by contradiction that the number of regular vertices is at least $k-1$. Then, by Lemma~\ref{lem:stwml_normal}, there are $2k-2$ edges between $Z$ and entering or exit vertices or regular vertices. We also have the edges $s_1z_1$ and $s_2z_2$. Moreover, every edge in $\delta_T(\mathcal{G}_1)$ and $\delta_T(\mathcal{G}_2)$ creates an edges in $\delta_T(Z)$ incident to irregular vertices by Lemma~\ref{lem:stwml_1edgetoZ} and the fact that $u_1u_2$ and $u_3u_4$ are the only irregular edges incident to each of these four vertices. Since there are two irregular edges, all the vertices of $Z$ have degree two and so $|\delta_T(Z)|=2k+2$. So $|\delta_T(\mathcal{G}_1)|+|\delta_T(\mathcal{G}_2)|=2$. But since one of $\mathcal{G}_1$ or $\mathcal{G}_2$ contains a vertex of degree three and no leaves, three edges have to leave it, a contradiction. So from now on we can assume that the number of regular vertices is $k-2$ and then all of $u_1,u_2,u_3,u_4$ are in $S$ (since $|S| \ge k+2$).

First assume that, $|\delta_T(\mathcal{G}_1)|=1$ or $|\delta_T(\mathcal{G}_2)|=1$, let us say wlog $\mathcal{G}_1$. Then, one vertex of the edge-gadget $\mathcal{G}_1$ is a leaf and $\mathcal{G}_2$ contains the vertex of degree three. Since there are two irregular edge-gadgets, all the vertices of $\mathcal{G}_1$ but the leaf have degree two in $T$. Moreover, since both $u_1$ and $u_2$ are in $S$, an entering or exit vertex incident to $u_1$ and $u_2$ have  to be of degree one in the restriction of $T$ to one of their edge-gadgets. 

We claim that it implies that an entering or exit vertex of both $u_1$ and $u_2$ in the edge-gadget of $\mathcal{G}_1$ have degree one in the restriction of $T$ to $\mathcal{G}_1$. Let us first prove that an edge of $\delta_T(\mathcal{G}_1)$ is incident to the entering or exit vertices of $u_1$, and that the same holds for $u_2$. Let us prove the statement for $u_1$ and assume by contradiction that it is not the case. Let $e'_i$ the be edge the closest of be the closest edge-gadget from $u_1u_2$ in the ordering of $u_1$ such that $x_{u_i}^{e'_i}$ or $y_{u_i}^{e'_i}$ has degree one in the graph restricted to $\mathcal{G}_{e'}$. Since $e'$ is regular, it implies by Lemma~\ref{lem:stwml_intersectionedgegadget} that an edge of $T$ is incident to the exit vertex of the gadget before $e'_i$ and the entering vertex of the gadget after $e'_i$. So an edge of $T$ leaving the gadget $\mathcal{G}_1$  is incident to entering or exit vertices of $u_1$, denoted by $x$. Now, since $\mathcal{G}_1$ contains one leaf and no vertex of degree three, if $x$ has degree one in $T$, its degree two incident local neighbor also is a leaf, a contradiction. So it has degree two and then has degree one in the gadget. A similar proof gives the same for $u_2$.

So one of the vertices $\{x_{u_1}^{u_1u2},y_{u_1}^{u_1u2} \}$ and one of the vertices $\{ x_{u_2}^{u_1u2},y_{u_2}^{u_1u2} \}$ have degree one in the subgraph $T'$ of $T$ induced by the vertices of $\mathcal{G}_1$. Since all the vertices but at most one (which cannot be a local vertex) have degree two in $T$ and $|\delta_T(\mathcal{G}_1)|=1$ by assumption, $T'$ is a Hamiltonian path on $\mathcal{G}_1$ between one vertex of $\{x_{u_1}^{u_1u2},y_{u_1}^{u_1u2} \}$ and one vertex of $\{ x_{u_2}^{u_1u2},y_{u_2}^{u_1u2} \}$, a contradiction with Lemma~\ref{clm:stwml_ham}. So we cannot have $|\delta(\mathcal{G}_1)|=1$.

So we can assume that $|\delta_T(\mathcal{G}_1)|=2$ and $|\delta_T(\mathcal{G}_2)|=2$. Let $\mathcal{G}_2$ be the edge-gadget containing a vertex of degree three and no leaves. Since it contains a branching node and no leaf, at least three edges are in $\delta_T(\mathcal{G}_2)$, a contradiction. 
\end{proof}

So the vertex cover $S(T)$ associated with every spanning tree $T$ with at most three leaves has size at most $k+1$. In order to prove that a spanning tree transformation provides a vertex cover transformation for the {\sf TAR} setting, we have to prove that, for every edge flip, then either $S$ is not modified, or one vertex is added to $S$ or one vertex is removed from $S$.

\begin{lemma}\label{lem:stwml_lastlemma}
Let $T_1$ and $T_2$ be two adjacent trees with at most three leaves. Then the symmetric difference between the sets $S$ associated with the two trees is at most one. 
\end{lemma}

\begin{proof}
We want to prove that $S(T_2)=S(T_1)$ or there exists $x$ such that $S(T_2) = S(T_1) \setminus \{x\}$ or $S(T_2) = S(T_1) \cup \{x\}$. In order to prove it, the rest of the proof is devoted to show that, if after some edge flip, a vertex is added to $S(T_2)$ then no vertex of $S(T_1)$ is removed in $S(T_2)$. We claim that it is enough to conclude. Indeed, since $|S| \le k+1$ by Lemma~\ref{lem:stwml_k+1} and $|S| \ge k$ (since $k$ is the minimum size of a vertex cover), if we want the symmetric difference to be at least two, then we must contain at least one vertex in $S(T_1) \setminus S(T_2)$ and conversely.
Let us now assume by contradiction that $|S(T_1) \setminus S(T_2)| = 1$ and $|S(T_2) \setminus S(T_1)| = 1$. Let $f$ be the edge of $T_1 \setminus T_2$ and $g$ be the edge of $T_2 \setminus T_1$. Let $u = S(T_2)\setminus S(T_1)$ and $v=S(T_1)\setminus S(T_2)$. Note that in order to modify $S(T)$ (for some tree $T$), we need to modify the degree of a special vertex in an edge-gadget of an edge of $G$ incident to it. So both $f$ and $g$ have to have both endpoints in the same edge-gadget. And the following remark ensures that the addition deletion of $f$ and $g$ can only modify by one vertex the set $S$. 
In particular, it implies that $|S(T_1) \, \triangle \, S(T_2)| \le 2$ ) 

\begin{remark}\label{rem:stwml_distance3}
Let $a,b$ be two special vertices in the same edge-gadget. The distance between $a$ and $b$ is at least three in $H(G)$.
\end{remark}

Remark~\ref{rem:stwml_distance3} ensures that, if we remove or add an edge of $T$, the degree of exactly one entering or exit vertex is modified. Since $S(T_2) \setminus S(T_1)$ and $S(T_1) \setminus S(T_2)$ are non empty, an entering or exit vertex of $u$ or $v$ has to be incident to $f$ and an entering or exit vertex of the other vertex of $u$ or $v$ has to be incident to $g$. By abuse of notation we will say that $f$ (resp. $g$) \emph{adds $u$ to $S(T_2)$} (resp. \emph{remove $v$ from $S(T_1)$}).

Since the edge $f$ (resp. $g$) adds $u$ or remove $v$, it has to have both endpoints in the same edge-gadget. Indeed, in order to add $u$ to $S(T_2)$ (or remove $v$ from $S(T_1)$) we must modify the degree of $x_v^e$ or $y_v^e$ (resp. $x_u^e$ or $y_u^e$) inside an edge-gadget.

Now let us distinguish cases depending on the degree of the endpoints of $f$. If both endpoints of $f$ are of degree two, then the deletion of $f$ creates two vertices of degree one. By Remark~\ref{rem:stwfl_nots1s2}, at most one of them is a leaf in $T_2$. So $g$ has to be incident to one of them. And by Remark~\ref{rem:stwml_distance3}, the edge $g$ cannot be incident to another special vertex of the edge-gadget. And thus $g$ does not add or remove a good vertex, a contradiction.

If one endpoint of $f$ has degree three and one has degree one, then the deletion of $f$ creates a vertex of degree zero. Thus $g$ must be incident to the degree zero vertex. Again, by Remark~\ref{rem:stwml_distance3}, $g$ cannot add or remove another vertex of $S(T_1)$, a contradiction.
Note that we get a similar contradiction if one endpoint of $f$ has degree two and the other has degree one.

So we can assume that one endpoint of $f$ has degree two and the other has degree three.
The edge $g$ cannot be added between two vertices of degree at least two in $T_1 \setminus f$ since otherwise $T_2$ would have two branching nodes. So at least one endpoint of $g$ (and even exactly one by Remark~\ref{rem:stwfl_nots1s2}) has degree one in $T_1 \setminus f$. By Remark~\ref{rem:stwml_distance3}, the endpoint of $g$ of degree one was already of degree one in $T_1$ since $g$ has to modify $S$. Moreover, the other endpoint of $g$ has degree exactly two in $T_1 \setminus f$ (otherwise we would create a vertex of degree four in $T_2$), and then by Remark~\ref{rem:stwml_distance3} has degree two in $T_1$. So in particular, the edge-gadget containing $f$ has one vertex of degree three and all the others have degree two and the edge-gadget containing $g$ has one vertex of degree one and all the others have degree two in $T_1$. 
Note that the deletion of $f$ can have two effects on $S(T_1)$: either a vertex disappears (because the degree of a special vertex drops from one to zero), or a vertex appears (because the degree of a special vertex drops from two to one).
\medskip

\noindent
\textbf{Case 1.} $v$ is removed from $S(T_1)$ when $f$ is removed. \\
Let $e=wv$ be the edge such that $f$ has both endpoints in $\mathcal{G}_e$. Let $R$ be the subgraph induced by the vertices of $\mathcal{G}_e$ and $T'$ the restriction of $T_1$ on $R$. Since $v$ is removed from $S(T_1)$, it implies that $x_v^e$ or $y_v^e$ have degree one in $T'$ and $f$ is incident to that vertex. Up to symmetry let us assume that it is $x_v^e$. If the edge $f$ is not $x_v^er_1$, then $r_1$ is a leaf of $T_1$, a contradiction since the degree one vertex has to be in the edge-gadget containing $g$. 
So the only edge of $T'$ incident to $x_v^e$ is $x_v^er_1$ and then $f=x_v^er_1$. Since one the two endpoints of $f$ has degree three in $T_1$ and $r_1$ has degree two in $H(G)$, there are two edges of $\delta(\mathcal{G}_e)$ incident to $x_v^e$.

\begin{claim}\label{clm:stwml_local2enter0}
Let $\mathcal{G}_e$ with $e=vw$ be an edge-gadget and $R$ be the subgraph of $H$ induced by the vertices of $\mathcal{G}_e$. There does not exist any tree $T$ such that, in the subgraph of $T$ induced by the vertices of $R$, all the local vertices but $r_1$ have degree two, $x_v^e$ has degree zero and $y_v^e$ has degree two.
\end{claim}

\begin{proof}
Let us denote by $T'$ the subgraph of $T$ induced by the vertices of $R$.
Since all the local vertices but $r_1$ have degree two and $y_v^e$ has degree two in $T'$, $T'$ contains the paths $r_1r_2$, $x_w^er_5r_6$ and $r_3r_4y_v^er_7r_8y_w^e$. Since $x_v^er_6$ is not an edge of $T$ (because we assumed that $x_v^e$ has degree zero in $R$) and $r_7$ does not have degree three, $r_6$ is a leaf of $T$, a contradiction.
\end{proof}

When $f$ is removed from $T_1$, $u$ is removed from $S$, thus $y_v^e$ has degree zero or two in $T'$. Since all the local vertices of $R$ have degree two in $T_1$ and $r_4$ has degree two in $H(G)$, both edges incident to it are in $T'$. And then $y_v^e$ does not have degree zero. So by Claim~\ref{clm:stwml_local2enter0}, the edge-gadget must contain another vertex of degree three or another leaf, a contradiction.
\medskip

\begin{figure}[bt]
    \centering
    \includegraphics[scale=.75]{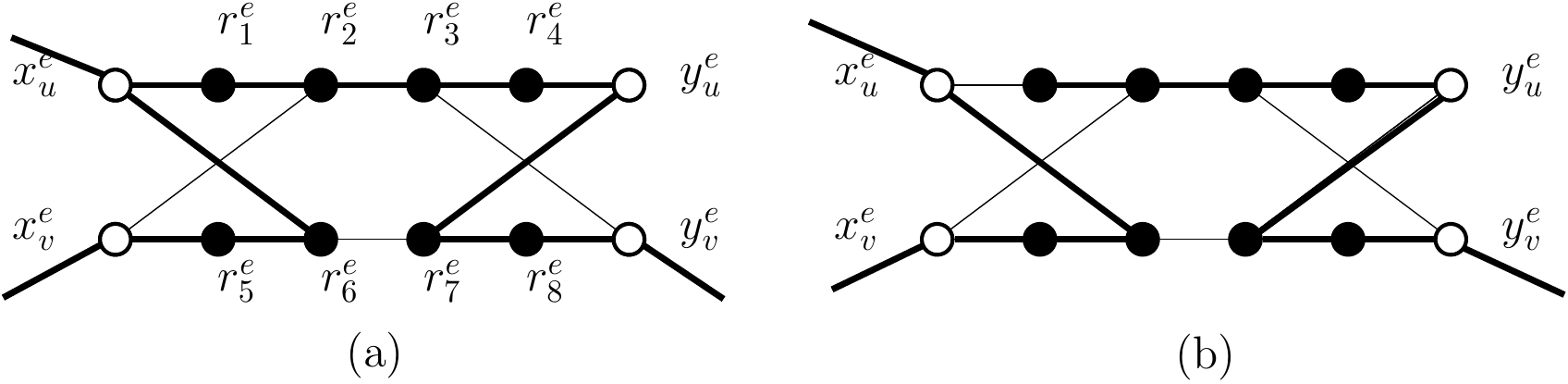}
    \caption{Illustration of the proof of Lemma~\ref{lem:stwml_lastlemma}.}
    \label{fig:stwml_lastlemma}
\end{figure}

\noindent
\textbf{Case 2.} $u$ is added to $S(T_1)$ when $f$ is removed.

Let $e=uv$ be the edge such that $f$ is in $\mathcal{G}_e$. Let $R$ be the subgraph induced by the vertices of $\mathcal{G}_e$. In that case, the vertices $x_u^e$ and $y_u^e$ have degree two in $R \cap T_1$ (if one of them was of degree one, $u$ was already in $S$ and none of them can be of degree zero, otherwise one local vertex should be a leaf, a contradiction since the leaf is in the edge-gadget containing $g$). Since all the local vertices have degree two or three, it implies that $x_v^er_5r_6x_u^er_1r_2$ and $y_v^er_8r_7y_u^er_4r_3$ are in $T_1$. But then $r_2r_3$ must be in $T_1$, otherwise they would be leaves of $T_1$ ($x^e_v r_2 \not\in E(T_1)$ otherwise $x^e_vr_5r_6x^e_ur_1r_2$ is a cycle (same for $r^e_3y^e_u$)). Moreover $r_6r_7$ cannot be an edge since otherwise there is a cycle (and both $r_6,r_7$ would have degree three). So the endpoint of $f$ of degree three has to be $x_u^e$ or $y_u^e$, w.l.o.g $x_u^e$.
So the graph around $\mathcal{G}_e$ is the graph represented in Figure~\ref{fig:stwml_lastlemma}(a). Note in particular that $|\delta_{T_1}(\mathcal{G}_{e})|=3$ since $x_{v}^{e}$ and $y_{v}^{e}$ have degree two in $T_1$ and $x_u^e$ has degree three in $T_1$.

Let $e'$ be the edge such that $g$ is in $\mathcal{G}_{e'}$ with $e'=u'v'$. Recall that $u'$ or $v'$ is removed from $S$ when $g$ is added. So we can assume without loss of generality that $g$ is incident to $x_{u'}^e$.
Let $R'$ be the subgraph induced by the vertices of $\mathcal{G}_{e'}$. All the vertices in the edge-gadget $\mathcal{G}_{e'}$ have degree two in $T_1$ but one vertex which has degree one. Moreover, $g$ is an edge between a vertex of degree two and a vertex of degree one. Since $u'$ is removed from $S$ when we add $g$, $x_{u'}^{e'}$ has degree exactly one in the restriction of $T_1$ to $R'$.

Let us first assume that $x_{u'}^{e'}r_1$ is in $T_1$ and then $x_{u'}^{e'}r_6$ is not in $T_1$ (i.e. $g = x^{e'}_ur_6$). Since all the local vertices but maybe $r_6$ have degree two and $y_{u'}^{e'}$ has degree two, all the subpaths $r_3r_4y_{u'}^{e'}r_7r_8y_{v'}^{e'}$, $x_{v'}^{e'}r_5r_6$ and  $x_{u'}^{e'}r_1r_2$ are in $T_1$. Since $r_3$ must have degree two in $T_1$ and $r_3y_{v'}^{e'}$ closes a cycle, $r_2r_3$ is in $T_1$. Since $x_{u'}^{e'}r_6$ is not in $T_1$ by assumption and $\mathcal{G}_{e'}$ does not contain any vertex of degree three, $r_6$ is a leaf of $T_1$ and then $|\delta_{T_1}(\mathcal{G}_{e'})|=3$ since $x_{u'}^{e'}$, $x_{v'}^{e'}$ and $y_{v'}^{e'}$ have degree two in $T_1$.

Let us now assume that $x_{u'}^{e'}r_1$ is not in $T_1$ and then $x_{u'}^{e'}r_6$ is (i.e.\ $g=x^{e'}_{u'}r_1$). Since all the local vertices but $r_1$ have degree two and $y_{u'}^{e'}$ has degree two, all the subpaths $r_3r_4y_{u'}^{e'}r_7r_8y_{v'}^{e'}$, $x_{v'}^{e'}r_5r_6x_{u'}^{e'}$ and  $r_1r_2$ are in $T_1$. Since $r_3$ must have degree two in $T_1$ and $r_3y_{v'}^{e'}$ closes a cycle, $r_2r_3$ is in $T_1$. Since $r_1$ is a leaf of $T_1$, $x_{u'}^{e'}$ has degree two in $T_1$. And then $|\delta_{T_1}(\mathcal{G}_{e'})|=3$ since $x_{u'}^{e'}$, $x_{v'}^{e'}$ and $y_{v'}^{e'}$ have degree two in $T_1$. The graph around $\mathcal{G}_{e'}$ is the graph represented in Figure~\ref{fig:stwml_lastlemma}(b).

So in both cases ($x_{u'}^{e'}r_1$ or $x_{u'}^{e'}r_6$ in $T_1$), we have $|\delta_{T_1}(\mathcal{G}_{e'})|=3$. Moreover, we have seen that $|\delta_{T_1}(\mathcal{G}_{e})|=3$.

We claim that $S(T_1)$ has size $k+1$. 
Recall that $f=uv$ and $g=u'v'$. Let us show that $S(T_1) \setminus \{ u \}$ is a vertex cover. For every edge $uw$ with $w \ne u',v$, since  $S(T_2)=(S(T_1) \cup \{u\}) \setminus \{u'\}$ is a vertex cover, $w$ is in $S(T_1)$. So if an edge is not covered in $S(T_1) \setminus \{ u \}$, it is $uu'$ or $uv$. After the edge flip, $x_u^{uv}$ and $y_u^{uv}$ have even degree in $T_2$ and thus $x_v^{uv}$ or $y_v^{uv}$ has degree one by Lemma~\ref{lem:stwml_S_VC}. Since neither $f$ nor $g$ changes the degree of $x_v^{uv}$ nor $y_v^{uv}$, $v \in S(T_1)$. So if an edge is not covered, it is $uu'$. But, since $u' \notin S(T_1)$, in the restriction of $T$ to $\mathcal{G}_{uu'}$, either $x_u^{uu'}$ or $y_u^{uu'}$ has degree one and this degree does not change after the edge flip, a contradiction since $u \notin S(T_2)$, so $uu'$ does not exist and then $S(T_1) \setminus \{ u \}$ is a vertex cover.
Since a minimum vertex cover has size $k$, $S(T_1)$ has size at least $k+1$ and then exactly $k+1$ by Lemma~\ref{lem:stwml_k+1}.

So $k-2$ vertices of $S(T_1)$ are not incident to any irregular edge-gadgets. (Indeed, there are at most four irregular vertices and $u' \notin S(T_1)$ is one of them. By Lemma~\ref{lem:stwml_normal}, this gives $2k-4$ edges in $\delta(Z)$. Since, for both $\mathcal{G}_e$ and $\mathcal{G}_{e'}$, there are three edges leaving the gadget and since $e$ and $e'$ are endpoint disjoint, this creates $6$ more edges incident to $Z$. Since there are moreover the two edges $s_1z_1$ and $s_2z_{k+1}$ in $\delta(Z)$. So in total, that gives $2k+4$ edges in $\delta(Z)$, a contradiction with the fact that all the vertices of $Z$ must have degree two. 
\end{proof}

Lemmas~\ref{lem:stwml_k+1} and~\ref{lem:stwml_lastlemma} immediately implies the following:

\begin{lemma}
If there is an edge flip reconfiguration sequence between two spanning trees $T_1$ and $T_2$, then there is a {\sf TAR}reconfiguration sequence (with threshold $k+1$) between $S(T_1)$ and $S(T_2)$.
\end{lemma}

\subsubsection{VCR to ST-reconfiguration}

We now prove the converse of the previous subsection\footnote{The statement will not be exactly the converse but it will actually be enough to conclude.}. We will prove that if there is a {\sf TJ}-transformation sequence  between two vertex covers $X$ and $Y$ then we also have an edge flip reconfiguration sequence between Hamiltonian paths corresponding to $X$ and $Y$. Let $X,Y$ be two vertex covers of size $k$. In the {\sf TJ}-adjacency rule, $X$ and $Y$ are adjacent if there exist two vertices $x$ and $y$ such that $Y= (X \setminus \{x\}) \cup \{y\}$. 

We have already remarked that there might be a lot of  Hamiltonian paths associated with a vertex cover $X$ in $H(G)$. Note that, in all these paths, for every $u \in X$, the subpath $P_x$ between the first entering vertex of $u$ and the last exit vertex of $u$ is the same. However {\em (i)} the order in which these subpath appear in the path may differ (depending in which ordering they are attached to $Z$); {\em (ii)} when we follow the path from $s_1$ to $s_2$ we might see the path in the ordering of $P_x$ or in the reverse ordering depending if the first vertex of $Z$ incident to $P_x$ is incident to the first entering vertex of the last exit vertex. The goal of the proof consists of showing that, if we have one of them, then we can reach all of them, i.e. change the order of appearance of the paths $P_x$ and reverse their ordering. The first part of this section consists of proving that they all are in the same connected component of the reconfiguration graph. Let us first show the following intermediate lemma.

\begin{lemma}\label{lem:stwml_bip}
Let $A,B$ be two sets such that $|A| = |B|+1$ and $G$ be the bipartite graph $\mathcal{B}$ on vertex set $(A,B \cup \{s_1,s_2\})$ where $A$ is complete to $B$ and $s_1,s_2$ be two vertices of $B$, each connected to exactly one (distinct) vertex of $A$. Let $P_1,P_2$ be two Hamiltonian paths with the same endpoints $s_1,s_2$. 
Then one can transform $P_1$ into $P_2$ via edge flips where all the intermediate spanning trees have at most three leaves.
\end{lemma}

\begin{proof}
We say that two paths $P,P'$ on the same vertex set {\em agree up to $i \in \mathbb{N}$} if the first $i$ vertices of  $P$ and $P'$ are the same. Note that $P_1,P_2$ agree up to $2$ since both start with $s_1$ and $s_1$ only have one neighbor in $\mathcal{B}$. We prove iteratively that if we have two paths that agree up to $i$, then we can transform the second into two paths that agree up to $i+1$.

Assume that $P_1$ and $P_2$ agree up to $i$. Let $u$ be the $i$-th vertex and $v$ be the $(i+1)$-th in $P_1$. If $v$ also is the $(i+1)$-th vertex in $P_2$, the conclusion holds. So we can assume that the $(i+1)$-th vertex of $P_2$ is $y \ne v$. Let $w$ be the vertex after $v$ in $P_2$. Note that it cannot be $y$ since both $y$ and $v$ are in the same set of $A,B$.
We perform the following edge flips in $P_2$: we remove $uy$ to create $uv$. We then remove $vw$ to create $yw$. 

After these two operations, all the vertices have degree two. Moreover the intermediate and final graphs are connected. Indeed, since $u,y,v$ appears in $P_2$ in that order, the removal of $uy$ to create $uv$ keeps a connected graph. And one can remark that the two operations just consists in permuting the subpath between $y$ and $v$ in $P_2$. To conclude, we have to prove that the edges we want to create indeed exist in $\mathcal{B}$. Since $B$ is complete to $A$, if $u$ and $w$ are in $B$, the conclusion follows. So we can assume that they are in $A \cup \{ s_1,s_2 \}$. Since $A$ is complete to $B$, and $u$ is distinct from $s_1$, and $y,v \in B$ (since they are not the last vertices of $P_1$ and $P_2$), the only edge that might not exist is $yw$ if $w = s_2$. But it is impossible since $s_2$ only have one neighbor in $\mathcal{B}$ and then the second to last vertex  of $P_1$ and $P_2$ are the same, i.e. $y$ cannot be incident to $u$ in $P_2$.
\end{proof}

Using this lemma, let us prove the following:

\begin{lemma}\label{lem:stwfl_associatedpaths}
Let $G=(V,E)$ be a graph and $X$ be minimum vertex cover of $G$. Then all the Hamiltonian paths associated with $X$ in $H(G)$ are in the same connected component of the reconfiguration graph of spanning trees with at most three leaves.
\end{lemma}

\begin{proof}
Let $k=|X|$. Let us denote by $A$ the set $Z$ of $H(G)$ and by $B$ the set $X$. Note that by construction $|A|=|B|+1$. We now add two new vertices $s_1,s_2$ one connected to $z_1$ and the other connected to $z_{k+1}$ and create all the edges between $A$ and $B$. We denote by $\mathcal{B}$ the resulting graph that satisfies the condition of Lemma~\ref{lem:stwml_bip}. Now one can associate with any Hamiltonian path associated with $X$ a path of $\mathcal{B}$ where $x$ in $\mathcal{B}$ is connected to $z,z'$ in $A$ if $z$ and $z'$ are the vertices of $Z$ attached to the first and last entering and exit vertices of $x$. By Lemma~\ref{lem:stwml_bip}, one can transform any path of $\mathcal{B}$ into any other. We claim that such a transformation can be immediately extended for the Hamiltonian paths of $H(G)$. Indeed, by definition of a Hamiltonian path of $H(G)$ associated with  $X$ the subpath $P_u$ between the first entering vertex of $u$ and the last exit vertex of $u$ (for $u \in X$) does not contain any other entering or exit vertex of vertices of $X$ and only contain degree two vertices. So the connectivity of the graph as well as its non-degree two vertices remain the same if can contract $P_u$ into a single vertex $u$. 

After this operation, we know that in the resulting Hamiltonian path, the subpaths associated with each vertex appear in the same ordering. However, it might be the case that in some path $z_i$ is connected to the first entering vertex $x$ of $u \in X$ and $z_{i+1}$ to the last exit vertex $y$ of $u$ and that we have the converse in the other path. In other words, instead of ``reading'' the path from the first entering vertex to the last exit vertex we ``read'' it in the other direction.
In that case, for every such $i$, we perform the following edge flips: remove $z_iu$ to create $z_iv$; and then remove $z_{i+1}v$ to create $z_{i+1}u$. 
\end{proof}

Let us now prove that if we are given any {\sf TJ}-transformation between two vertex covers $X$ and $Y$ can be adapted into an edge flip transformation between the corresponding Hamiltonian paths via spanning trees of at most three nodes. In order to prove it, we simply have to prove that we can do it for each single step transformation.

\begin{lemma}\label{lem:stwml_vctoham}
Let $X$ be a minimum vertex cover of $G$ and $Y= (X \setminus \{ u \}) \cup \{ v \}$ be another vertex cover, for some vertices $u$ and $v$. Then we can transform any Hamiltonian path associated with $X$ into any Hamiltonian path associated with $Y$ via a sequence of spanning trees with at most three leaves.
\end{lemma}

\begin{proof}
By Lemma~\ref{lem:stwfl_associatedpaths}, all the Hamiltonian associated with $X$ are in the same connected component of the reconfiguration graph and the same holds for $Y$. So we simply have to show that there exists a transformation from a Hamiltonian path associated with $X$ into a Hamiltonian path associated with $Y$. First, observe that since $X$ and $Y$ are both minimum vertex covers of $G$ and $Y = (X \setminus \{u\}) \cup \{v\}$, $X \setminus \{u\}$ covers all the edges of $G$, but $uv$. In particular, all the neighbors of $u$ but $v$ are in $X$. Similarly, all the neighbors of $v$ but $u$ are in $Y$. Let $W= X \cap Y$  given with an arbitrary ordering of $W$. The \emph{canonical path associated with $W,u$} (resp. $W,v$) is the Hamiltonian path of $H(G)$ with the ordering $u$ (resp. $v$) and then the ordering of $W$. More formally, recall that given a vertex cover $W$, we can define a path $P_w$ for every $w \in W$ between the first entering vertex of $w$ and the last exit vertex of $w$ that does not contain any special vertex of $w' \in W$ with $w' \ne w$. And that any Hamiltonian path associated with $W$ is the concatenation of these paths linked together thanks to the vertices of $Z$. So the ordering of $W$ of a path $P$ is the ordering of appearance of the subpaths $P_w$ for $w \in W$. In particular, in the ordering of $W_u$, the subpath $P_u$ appears at the beginning of the path and then $P_u$ is connected to $z_1$ and $z_2$.

The {\em half-path $T_h$ associated with $W,u,v$} is the following.
For every edge-gadget $\mathcal{G}_e$ with $e$ distinct from the first edge of $v$, the restriction of $T_h$ around $\mathcal{G}_e$ is one of the graphs of Figure~\ref{fig:stwl_reduction2}; If both endpoints of $e$ are in $W \cup \{u,v \}=X \cup Y$, the gadget is the one of Figure~\ref{fig:stwl_reduction2}(a), otherwise it is the one of Figure~\ref{fig:stwl_reduction2}(b) (the edges of $\delta_{T_h}(\mathcal{G}_e)$ being incident to the entering and exit vertex of $W \cup \{u,v\}$). For the edge-gadget of $e'=vw$ (note that we possibly have $w=u$), the first edge of the ordering of $v$, the restriction of $T_h$ around $\mathcal{G}_{e'}$ is the graph $x_v^{e'}r_1r_2r_3r_4y_v^{e'}$ and $x_w^{e'}r_5r_6r_7r_8y_w^{e'}$ plus edges leaving $y_v^{e'}$, $x_w^{e'},y_w^{e'}$ but no edge leaving $x_v^{e'}$.

Let us now explain how the vertices of $Z$ are connected to entering and exit vertices. The vertex $z_1$ is incident to $s_1$ and the first vertex of $u$. The vertex $z_2$ is incident to the last exit vertex of $u$ and the last exit vertex of $v$. Moreover, the vertex $z_{i+1}$ is incident to the last exit vertex of the $i$-th vertex of $W$ and the first entering vertex of the $(i+1)$-th vertex of $W$. Finally the vertex $z_{k+1}$  is incident to $s_2$.

One can easily check that the following holds for $T_h$:
\begin{itemize}
    \item All the vertices of $T_h$ have degree two but $z_2$ that has degree three\footnote{$z_2$ is incident to the last exit vertex of $u$, the last exit vertex of $v$ and the first entering vertex of the first vertex of $W$.} and the first entering vertex of $v$ that has degree one\footnote{It is not connected to any vertex of $Z$.}.
    \item The subpath of $T_h$ from $s_1$ to $z_2$ is $s_1z_1$ and then the concatenation of the paths (for every edge $e$ incident to $u$) $x^e_u,r_1,r_2,r_3,r_4,y^e_u$ (or $x^e_u,r_5,r_6,r_7,r_8,y^e_u$) connected by the special edges between consecutive edge-gadgets of $u$. Indeed, $W \setminus \{u\}$ covers all the edges but $uv$, all the neighbors of $u$ but $v$ are in $W$. Let $e''=uv$. The construction of $W,u,v$ also ensures that $\mathcal{G}_{e''}$ contains the subpath $x^{e''}_ur_1r_2r_3r_4y^{e''}_u$ (ou $r_5r_6r_7r_8$), no matter whether $e''$ is the first edge of $v$, or not.
    \item Similarly the subpath of $T_h$ from the first entering vertex of $v$ to $z_2$ is the concatenation of the paths (for every edge incident to $v$) containing the entering vertex of $v$ for the current edge, $r_1,r_2,r_3,r_4$ and the exit vertex of $v$.
    \item The subpath of $T_h$ from $z_2$ to $s_2$ is the subpath of the canonical path associated with $W,u$ except that for every edge $e=vw$, the graph around $\mathcal{G}_e$ is the graph of Figure~\ref{fig:stwl_reduction2}(a) instead of the graph of Figure~\ref{fig:stwl_reduction2}(b). 
\end{itemize}
In particular, one can notice that $T_h$ is a tree. Note moreover that, if we denote by respectively $e$ and $e'$ the first edges of $u$ and $v$ respectively, the edge flip $z_1x_u^e$ into $z_1x_v^{e'}$ transforms the half-path of $W,u,v$ into the half-path of $W,v,u$.

So in order to conclude, we simply have to prove that we can transform the canonical path associated with $W,u$ into the half-path associated with $W,v,u$. 

The proof is based on local transformations for every edge-gadgets iteratively on the gadgets. There are three types of local transformations illustrated in Figure~\ref{fig:stwml_transformation1}.

\begin{figure}[bt]
    \centering
    \includegraphics[scale=.75]{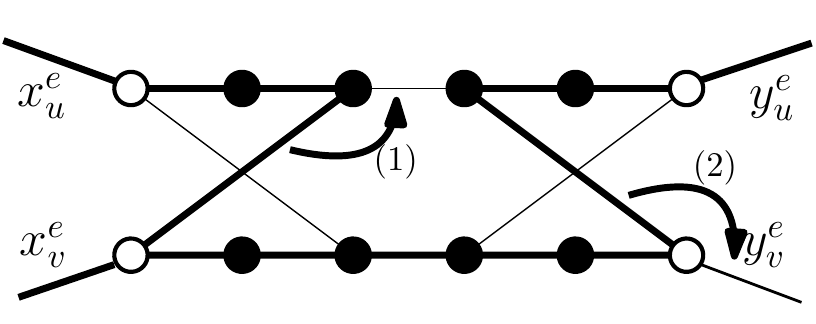}
    \caption{First transformation of an edge-gadget.}
    \label{fig:stwml_transformation1}
\end{figure}

Let $P$ be the Hamiltonian path associated with $X,u$. Since $X$ and $Y$ are vertex covers, the path $P$ has the following properties:

\begin{itemize}
    \item The restriction of $P$ to every edge-gadget with endpoint $v$ is of type Figure~\ref{fig:stwl_reduction2}(b). Indeed $v \notin X$ and $X$ is a vertex cover. In particular, the restriction of $P$ to the edge-gadget of $uv$ is of type Figure~\ref{fig:stwl_reduction2}(b). 
    \item The restriction of $P$ to every other edge-gadget edge incident to $u$ is of type Figure~\ref{fig:stwl_reduction2}(a). Indeed $(X \setminus \{ u \}) \cup \{ v \}$ is a vertex cover. So all the neighbors of $u$ but $v$ are in $X$.
\end{itemize}

Let us denote by respectively $x_u$, $x_v$, $y_u$, $y_v$ the first entering vertices of $u$ and $v$ and the last exit vertices of $u$ and $v$. Since $P$ is associated with $X,u$, $z_1x_u$ and $z_2y_u$ are edges of $P$. Let us delete $z_1x_u$ and add $z_1x_v$. Note that this operation creates a vertex of degree one (namely $x_u$) and a vertex of degree three (namely $x_v$). The resulting graph is indeed connected since only $s_1z_1$ is attached to $z_1$.

Let us prove iteratively on the edges incident to $v$ that we can transform the current graph keeping the degree sequence and the connectivity in such a way {\em (i)} the unique vertex of degree three is the current entering vertex $x$ of $v$, and {\em (ii)} there is a subpath attached to $x$ which is $s_1z_1$ and then the concatenation of the paths (for every edge incident to $v$ smaller than the current edge) containing the entering vertex of $v$ for that edge, $r_1,r_2,r_3,r_4$ and the exit vertex of $v$.

Note that the property indeed holds at the beginning since $x_v$ the first entering vertex of $v$ has degree three and there is a path $s_1z_1$ attached to it. Since $v$ is not in $X$, the graph around the current edge $e$ is indeed the graph of Figure~\ref{fig:stwl_reduction2}(b) in $P$. So in the current tree, we have the graph of Figure~\ref{fig:stwml_transformation1}. One can remark that the transformations proposed in Figure~\ref{fig:stwml_transformation1} keeps the degree sequence. Moreover, after these operations, one can note that the property holds up to the next entering vertex of $v$. So we simply have to show that the connectivity is kept to conclude. The first transformation indeed keeps connectivity. The second also keeps connectivity since the next entering vertex of $v$ is not in the subpath between $s_1$ and $x_v^e$. 

When we treat the last edge-gadget of $v$, we simply have to connect the last exit vertex to $z_2$ (which now has degree three) in order to obtain the subpath associated with $W,v,u$.
Similarly, we can transform the path associated with $W,v$ into the half-path associated with $W,u,v$. And, as we already observed, there is one edge flip that transforms the first into the second.
So it is possible to transform the canonical path associated with $W,u$ into the canonical path associated with $W,v$, which completes the proof.
\end{proof}

\subsection{Spanning trees with more leaves}\label{sec:stwfl_moreleaves}

Note that the reduction given in the previous Section can be easily adapted to more leaves.

\begin{theorem}
 Let $k \ge 3$ be an integer. \STWkL is PSPACE-complete.
\end{theorem}

\begin{proof}
We perform the same reduction as in the previous sections except that in the construction of the graph $H(G)$ we replace the two vertices $s_1,s_2$ by the $\ell+1$ additional vertices $s_1,s_2,\ldots,s_{\ell+1}$ where $s_1$ is connected to $z_1$, $s_2$ is connected to $z_{k+1}$ and $s_3,\ldots,s_{\ell+1}$ are connected to $s_1$. Note that in any spanning tree, $s_2,\ldots,s_{\ell+1}$ are leaves. So Remark~\ref{rem:stwfl_nots1s2} also holds with this reduction. The rest of the proof is the same.
\end{proof}

\section{Spanning tree with many leaves}

Before stating the main results of this section, let us prove the following:

\begin{lemma}\label{lem:stwml_samecomp}
Let $G$ be a graph and $T_1,T_2$ be two trees. There exists a transformation from $T_1$ to $T_2$ such that every intermediate tree $T$ satisfies $in(T) \subseteq in(T_1) \cup in(T_2)$. \\
In particular, all the trees with the same set of internal nodes are in the same connected component of the reconfiguration graph.
\end{lemma}

\begin{proof}
Let us prove that we can iteratively add an edge of $E(T_1) \setminus E(T_2)$ to $T_2$ and remove an edge of $E(T_2) \setminus E(T_1)$  without creating any internal node in $V \setminus (in(T_1) \cup in(T_2)$. Let $uv \in E(T_2) \setminus E(T_1)$. We add this edge to $T_1$ and observe that it creates a unique cycle in $T_1$. If it does not create any internal node note in $V \setminus (in(T_1) \cup in(T_2)$, we remove from the cycle any edge that is not in $T_1$. Otherwise, assume $u \notin in(T_1) \cup in(T_2)$. In particular $u$ is a leaf of $T_1$ and $uv$ is an edge of $T_1$ so $v \in in(T_1)$. Since $u$ was a leaf of $T_2$, the cycle in $T_2 \cup \{uv \}$ passes through the other edge incident to $u$. We remove it in order to keep a connected graph.
\end{proof}

\subsection{Hardness results}

\begin{theorem}\label{thm:stwml-split-bipartite}
 \STWML is PSPACE-complete even restricted to bipartite graphs or split graphs.
\end{theorem}

\begin{proof}
We first prove Theorem \ref{thm:stwml-split-bipartite} for bipartite graphs and then explain how we can adapt the proof for split graphs. We give a polynomial-time reduction from the {\sf TAR}-\textsc{Dominating Set Reconfiguration} problem (abbreviated in {\sf TAR}-DSR problem). Haddadan et al.~\cite{Haddadan16} showed that the {\sf TAR} reconfiguration of dominating sets is PSPACE-complete. More precisely, they proved that given a graph $G$ and $D_\source$, $D_\target$ two dominating sets of $G$, deciding whether there is a reconfiguration sequence between $D_\source$ and $D_\target$ under the {\sf TAR($\max(|D_\source|,|D_\target|)+1$)} rule is PSPACE-complete.

Let $G=(V,E)$ be a graph with vertex set $V(G) = \{v_1,v_2, \ldots, v_n\}$ and let $D_\source$, $D_\target$ be two dominating sets of $G$. Free to add vertices to the set of smallest size, we can assume without loss of generality that $D_\source$ and $D_\target$ are both of size $k$. Let $(G,k+1,D_\source,D_\target)$ be the corresponding instance of {\sc Dominating Set Reconfiguration} under \TAR, where $k+1$ is the threshold that we cannot exceed. We construct the bipartite graph $G'$ as follows: we make a first copy $A=\{a_1,a_2, \ldots, a_n\}$ of the vertex set of $G$, and a second copy $B=\{b_{1,0},b_{1,1},b_{2,0},b_{2,1},\ldots,b_{n,0},b_{n,1}\}$ where we double each vertex.  We add an edge between $a_i \in A$ and $b_{j,k} \in B$ for $k \in \{0,1\}$ if and only if $v_j \in N_G[v_i]$. Note that $N(b_{i,0})=N(b_{i,1})$, for every $1 \le i \le n$. We finally add a vertex $x$ adjacent to all the vertices in $A$ and we attach it to a degree-one vertex $y$. See Figure \ref{fig:reduc-bipartite} for an illustration. Note that $G'$ is bipartite since $A \cup \{y\}$ and $B \cup \{x\}$ induce two independent sets. 

\begin{figure}[bt]
    \centering
    \begin{subfigure}[b]{0.45\textwidth}
        \vfill
        \centering
        \includegraphics[height=4.5cm]{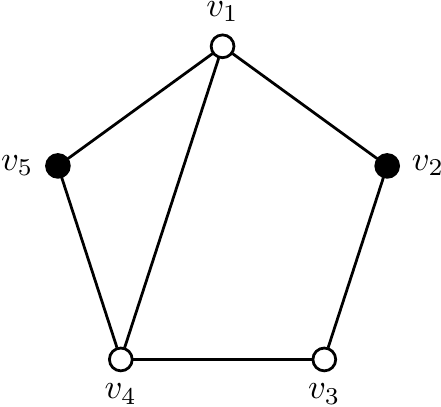}
        \vfill
        \caption{Original graph $G$.} \label{subfig:original-bipartite}
    \end{subfigure}
    \begin{subfigure}[b]{0.45\textwidth}
        \centering
        \includegraphics[height=8cm]{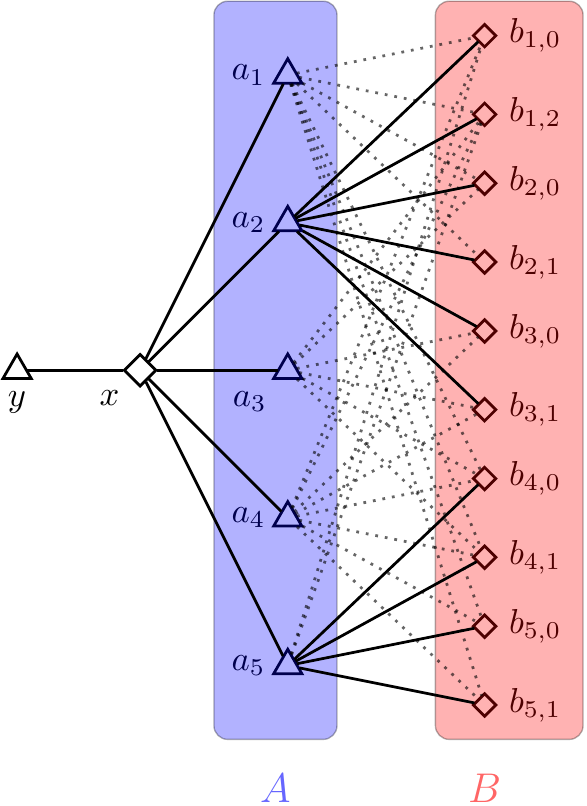}
        \caption{Corresponding bipartite graph $G'$.} \label{subfig:reduc-bipartite}
    \end{subfigure}
    \caption{Example for the reduction of Theorem \ref{thm:stwml-split-bipartite}: the dominating set $D = \{v_2,v_5\}$ of $G$ is depicted by the black vertices and the spanning tree of $G'$ associated with $D$ is the tree induced by the solid edges. For the split case, we add all the possible edges in $G'[A]$ so that $G'[A \cup \{x\}]$ is a clique and $G'[B \cup \{y\}]$ an independent set.}
    \label{fig:reduc-bipartite}
\end{figure}

\begin{claim}\label{clm:stwml:hardbip}
For every spanning tree $T$ of $G'$, $in(T) \cap A$ is a dominating set of $G$.
\end{claim}

\begin{proof}
Let $b_i$ be a vertex of $B$. Since $x$ is an internal node of $T$, there is a path from  $b_i$ to $x$. Since $N(b) \subseteq A$, the second vertex of the path is in $A$. So there exists an internal node of $T$ incident to $b_i$.
\end{proof}

\begin{claim}\label{clm:stwml:hardbip2}
For every spanning tree $T$ of $G'$, there exists a tree $T_A$ in the connected component of $T$ such that $in(T_A) \subseteq in(T) \cap (A \cup \{ x \})$.
\end{claim}

\begin{proof}
If $in(T) \subseteq A \cup \{ x \}$, the conclusion holds. So we can assume that there exists $b \in B$ such that 
$b \in in(T)$. Let us prove that we can transform $T$ into another spanning tree $T'$ such that $in(T') \subseteq in(T) \setminus \{b\}$ without creating a new internal node.
First, recall that $x \in in(T)$ since $x$ must be an internal node in any spanning tree of $G$. Let $a$ be the unique neighbor of $a$ in the path from $b$ to $x$ in $T$. Now, for every vertex $a' \ne a$ incident to $b$, we remove the edge $ba'$ and create the edge $xa'$. Since $x$ is internal in every tree, it does not increase the number internal nodes. Since $b$ is on the path between $a'$ and $x$ in $T$, it keeps the connectivity of the graph. After all these operations, the resulting tree $T_A$ satisfies $in(T_A) \subseteq in(T) \setminus \{ b \}$. We repeat this operation until no vertex of $B$ is internal.
\end{proof}

Let $D$ be a dominating set of $G$ of size $k$. We can associate with $D$ a spanning tree of $G'$ with $k+1$ internal nodes as follows. We attach every vertex in $A \cup \{y\}$ to $x$. Every vertex $b_i \in B$ is a leaf adjacent to a vertex that dominates $v_i$ in $D$. If $v_i$ has more than one neighbor in $D$, we choose the one with the smallest index. This spanning tree is called the \emph{spanning tree associated with $D$}.
See Figure \ref{fig:reduc-bipartite} for an example.

Let $(G,k+1,D_\source,D_\target)$ be an instance of {\sf TAR}-DSR. It is clear that $G'$ can be constructed in polynomial-time as well as $T_\source$ and $T_\target$ the spanning trees associated with $D_\source$ and $D_\target$. 
It remains to prove that $(G,k+1,D_\source,D_\target)$ is \textsf{yes}-instance of {\sf TAR}-\textsc{DSR} if and only if $(G',k',T_\source,T_\target)$ is a \textsf{yes}-instance of \STWML. 
\medskip

\noindent
($\Leftarrow$)
Suppose that there is a reconfiguration sequence of spanning trees $S' = \langle T_0 = T_\source, T_1, \ldots, T_{\ell'}=T_\target \rangle$, where each spanning tree has at most $k+2$ internal nodes. Since $x$ is an internal node of any spanning tree of $G'$, $D_i=in(T_i) \cap A$ has size at most $k+1$, for every $0 \le i \le \ell'$. Moreover, by construction of $T_\source$ and $T_\target$, $in(T_\source) \cap A=D_\source$ and  $in(T_\target) \cap A=D_\target$. For every vertex $b$ of $B$ and every $i$, there exists a vertex of $A$ in the path from $b$ to $x$ in $T_i$. It follows that the set $D_i$ is a dominating set of $G$, for every $0 \le i \le \ell'$. Hence, $\langle D_\source=D_0,\ldots,D_{\ell'}=D_\target \rangle$ is a transformation from $D_\source$ to $D_\target$.
It remains to prove that $|D_{i+1} \, \triangle \, D_i| \le 1$ for every $0 \le i < \ell'$ to guarantee the existence of a {\sf TAR}($k+1$)-reconfiguration sequence between $D_\source$ and $D_\target$ in $G$. What we will show is actually a bit more subtle. We will show that it is not necessarily the case but that, if it is not the case for some $i$, there exists a dominating set $D_i'$ such that $D_i,D_i',D_{i+1}$ satisfies $|D_i \, \triangle \, D_i'| \le 1$ and $|D_i' \, \triangle \, D_{i+1}| \le 1$ which is enough to conclude.

We consider an edge flip between two consecutive spanning trees of $S'$, let us say $T_i$ and $T_{i+1}$. Let $e_i$ (respectively $e_{i+1}$) be the edge in $E(T_i) \setminus E(T_{i+1})$ (resp. $e_{i+1} = E(T_{i+1}) \setminus E(T_i)$). We denote by $e_i \rightsquigarrow e_{i+1}$ the edge flip that transforms $T_i$ into $T_{i+1}$. Since $G'$ is bipartite and $y$ has degree one in $G'$, both $e_i$ and $e_{i+1}$ have an endpoint in $A$ and $|\{e_i,e_{i+1}\} \cap A| \le 2$. Hence, $|D_i \, \triangle \, D_{i+1}| \le 2$. If $e_i$ and $e_{i+1}$ are incident to a same vertex of $A$, the edge flip preserves its degree and thus $|D_i \, \triangle \, D_{i+1}| = 0$. Let us denote by $a_i$ (resp. $a_{i+1}$) the vertex in $A$ incident to $e_i$ (resp. $e_{i+1}$) in $A$.  Observe that $|D_i \, \triangle \, D_{i+1}| \le 1$ unless $a_i$ has degree two and $a_{i+1}$ is a leaf in $T_i$. 

First assume that the other endpoint of $e_i$ is a vertex $b_i$ in $B$. Let $b_i'$ be the vertex $b_{j,1}$ if $b_i$ is $b_{j,0}$ or $b_{j,0}$ if $b_i$ is $b_{j,1}$, i.e.\ $b_i'$ corresponds to the false twin of $b_i$. We claim that there exists an internal node of $T_i$ distinct from $a_i$ incident to 
$b_i'$. By contradiction. The neighbor of $b_i'$ on the unique path from $b_i'$ to $x$ has to be $a_i$ (since otherwise the neighbor of $b_i'$ which is not $x$ has to have a path to $x$ which provides the desired internal node).
Since $a_i$ has degree two, the two neighbors of $a_i$ are $b_i$ and $b_i'$. But this $P_3$ has to be connected to $x$. So $b_i$ or $b_i'$ are incident to another internal vertex $a'$ of $A$. 
But then $D_i \setminus \{a_i\}$ is a dominating set (since $a'$ dominates $b_i$) and then setting $D_i'=D_i \setminus \{a_i\}$, we have  $|D_i' \, \triangle \, D_i| \le 1$ and $|D_{i+1} \, \triangle \, D_i'| \le 1$.

Now assume that $e_i=a_ix$. Let $b_i$ be the other neighbor of $a_i$ in $T_i$.  Let $b_i'$ be the vertex $b_{j,1}$ is $b_i$ is $b_{j,0}$ or $b_{j,0}$ if $b_i$ is $b_{j,1}$. The neighbor $a$ of $b_i'$ on the path from $b_i'$ to $x$ is neither $a_i$ nor $a_{i+1}$. So, again $D_i':=D_i \setminus \{a_i\}$ is a dominating set and we have  $|D_i' \, \triangle \, D_i| \le 1$ and $|D_{i+1} \, \triangle \, D_i'| \le 1$.
So the conclusion follows.

\medskip

\noindent
($\Rightarrow$)
Suppose now that there exists a {\sf TAR}($k+1$)-reconfiguration sequence $S' = \langle D_0 = D_\source, D_1, \ldots, D_{\ell'}=D_\target \rangle$, from $D_\source$ to $D_\target$ in $G$. Let us prove that, for every $i$, there exists an edge flip between a tree with internal node $\{x\} \cup A(D_i)$ and a tree with internal nodes included in $\{x\} \cup A(D_{i+1})$ (for every $j$, $A(D_j)$ is the set of vertices of $A$ corresponding to the set $D_j$). The existence of a transformation from $T_\source$ to $T_\target$ follows since all the trees with the same set of internal nodes of $A$ are in the same connected component of the reconfiguration graph by Lemma~\ref{lem:stwml_samecomp} and Claim~\ref{clm:stwml:hardbip2}. Free to permute $D_i$ and $D_{i+1}$, we can assume that $D_i$ contains $D_{i+1}$. Let $u$ be the vertex of $D_i \setminus D_{i+1}$. Now, let $T_i$ be a spanning tree with internal nodes included in $D_i \cup \{x \}$. By Claim~\ref{clm:stwml:hardbip2}, we can assume that the set of internal nodes of the spanning tree is included $\{ x \}\cup A(D_i)$. Now, we can flip the edges in such a way, all the edges incident to $x$ are in the spanning tree (since $x$ already is an internal node, it cannot increase the number of internal nodes). Now for every edge $a_ub$ in the spanning tree, we can replace it by an edge $a_vb$ with $v \in D_{i+1}$ since $D_{i+1}$ is a dominating set. After this transformation, the resulting tree has its internal nodes in $x \cup A(D_{i+1})$ which completes the proof.

We now discuss how to adapt the proof for split graphs. First, we add an edge between any two vertices in $A$ so that $G'[A]$ is a clique. Then, observe that $G'[A \cup \{x\}]$ is a clique, and $G'[B \cup \{y\}]$ an independent set. Given two dominating sets $D_\source$ and $D_\target$, we associate with $D_\source$ and $D_\target$ the two corresponding spanning trees $T_\source$ and $T_\target$ of $G'$ in the same way as in the proof for bipartite graphs. Now, given a reconfiguration sequence between $D_\source$ and $D_\target$, the same proof as for bipartite graphs also holds here. A transformation for bipartite graphs indeed gives a transformation for split graphs. The converse direction also holds since we can assume that no vertex of $B$ is internal all along the transformation.
Suppose now that there is a reconfiguration sequence $S'$ between $T_\source$ and $T_\target$. We can assume that every vertex in $B$ is a leaf in any spanning tree $T_i$ of $S'$ for the same reason as in the proof for bipartite graphs. Since $x$ must be an internal node in any spanning tree, we can suppose that no edge between two vertices in $A$ is added to $S'$. Suppose that an edge $a_ia_j$ is added. This edge must have replaced either the edge $xa_i$ or the edge $xa_j$. In any way, we cannot decrease the number of internal nodes since $x$ is still an internal node. It follows that $S'$ only touches edges between $A$ and $B$. Hence, we can conclude as in the proof for bipartite graphs. The conclusion follows.
\end{proof}

Using another reduction we can prove the following: 

\begin{theorem}\label{thm:stwml-planar}
\STWML is PSPACE-complete on planar graphs.
\end{theorem}

\paragraph*{The reduction.} 
First, observe that {\sc MVCR} is PSPACE-complete, even if the input graph is planar~\cite{Hearn:2005:PSP:1140710.1140715}\footnote{Actually, Hearn and Demaine \cite{Hearn:2005:PSP:1140710.1140715} showed the PSPACE-completeness for the reconfiguration of maximum independent sets. Since the complement of a maximum independent set is a minimum vertex cover, we directly get the PSPACE-completeness of MVCR.}. 
We use a reduction from MVCR, which is a slight adaptation of the reduction used in \cite[Theorem 4]{DBLP:conf/mfcs/MizutaHIZ19}. Let $G=(V,E)$ be a planar graph and let $(G, C_\source, C_\target)$ be an instance of MVCR. We can assume that $G$ is given with a planar embedding of $G$ since such an embedding can be found in polynomial time. Let $F(G)$ be the set of faces of $G$ (including the outer face). We construct the corresponding instance $(G', k, T_\source, T_\target)$ as follows (see Figure \ref{fig:reduc-planar} for an example).

We define $G'$ from $G$ as follows. We start from $G$ and first subdivide every edge $uv \in E(G)$ by adding a new vertex $w_{uv}$. Then, for every face $f \in F(G)$, we add a new vertex $w_f$ adjacent to all the vertices of the face $f$. Finally, we attach a leaf $u_f$ to every vertex $w_f$. Note that $G'$ is a planar graph and $|V(G')| = |V(G)|+|E(G)|+2 \cdot |F(G)|$. The vertices $w_{uv}$ for $uv \in E$ (resp. $w_f$ for $f \in F$) are {\em edge-vertices} (resp. {\em face-vertices}). The vertices $u_f$ for every $f$ are called the \emph{leaf-vertices}. Note that, for every spanning tree $T$, all the face-vertices are internal nodes of $T$ and all the leaf-vertices are leaves of $T$. The vertices of $V(G')$ which are neither edge, face of leaf vertices are called {\em original vertices}.
Finally, we choose arbitrarily order of $V(G)$ and $F$. It will permit us to define later a canonical spanning tree for every vertex cover. 

\begin{figure}[bt]
    \centering
    \begin{subfigure}[b]{0.49\textwidth}
        \centering
        \includegraphics[height=6.3cm]{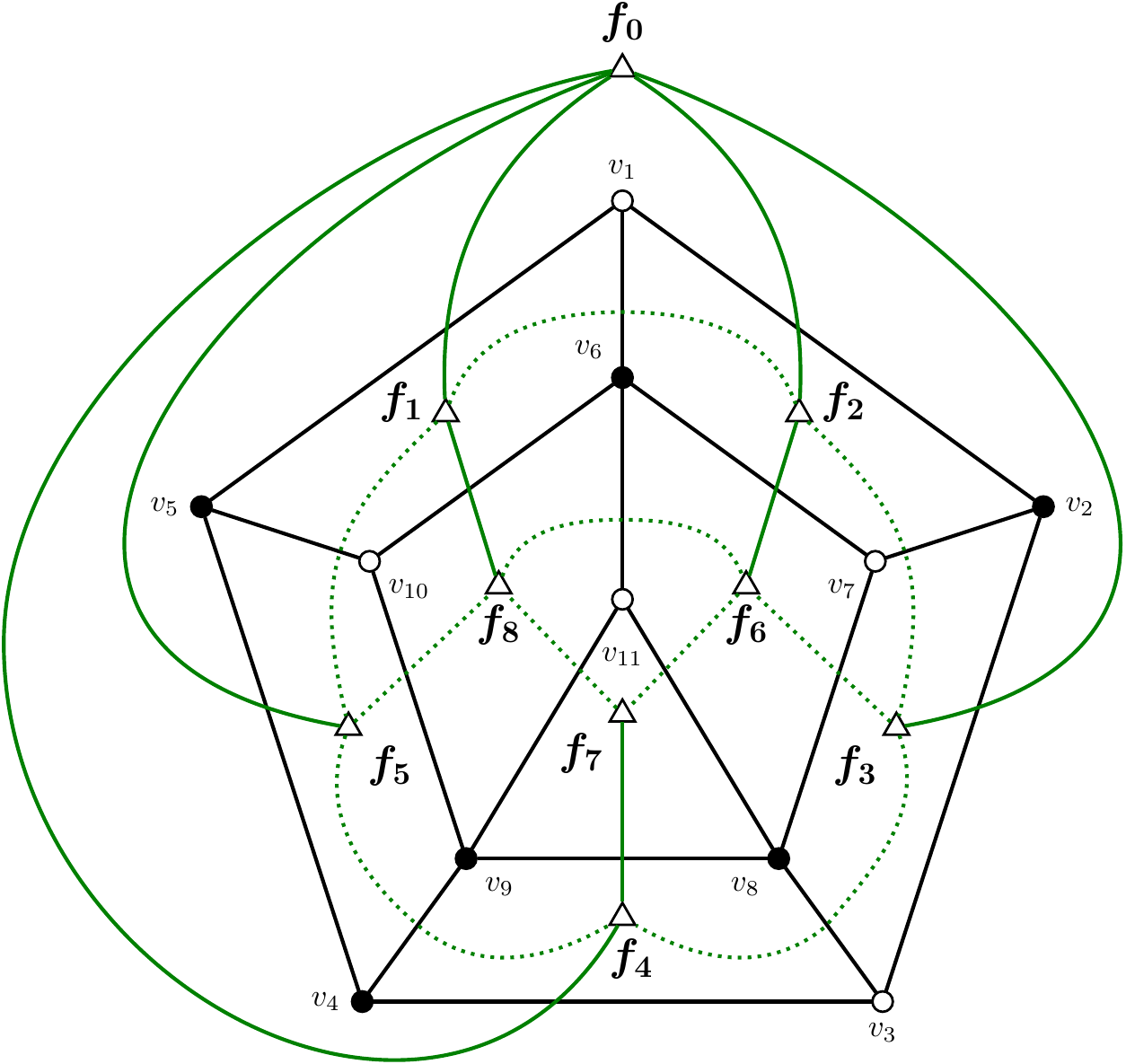}
        \caption{Original labeled planar graph $G$. } \label{subfig:original-planar}
    \end{subfigure}
    \begin{subfigure}[b]{0.49\textwidth}
        \centering
        \includegraphics[height=6.3cm]{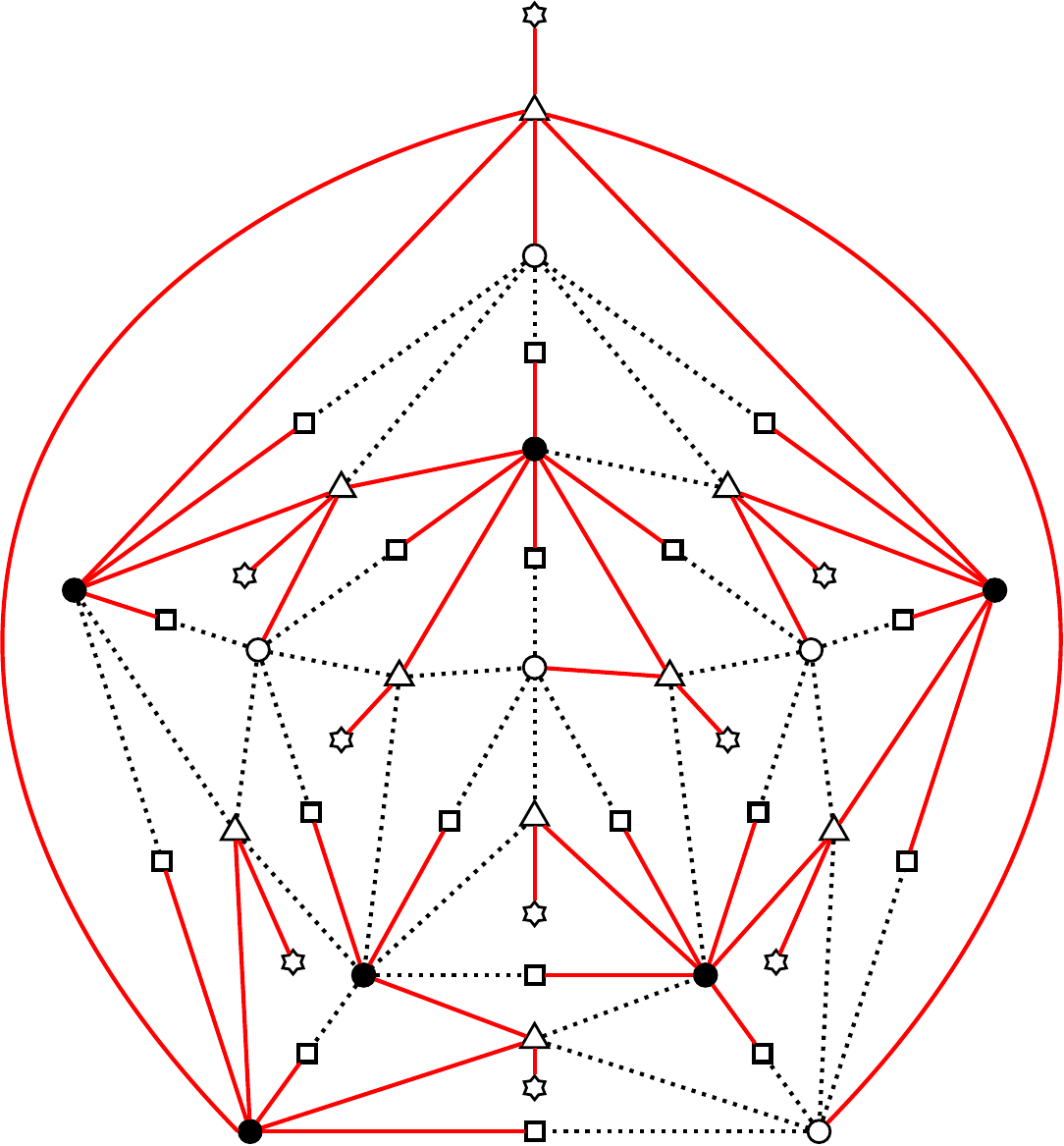}
        \caption{Corresponding planar graph $G'$.} \label{subfig:reduc-planar}
    \end{subfigure}
    \caption{Reduction for Theorem \ref{thm:stwml-planar}. The vertex cover $C$ of $G$ is depicted by the black vertices. The dual graph is the graph induced by the green edges. The spanning tree obtained from the BFS is represented by the solid edges.
    The face-vertices (respectively edge-vertices) of $G'$ are depicted by triangles (resp. squares). The spanning tree $T$ of $G'$ associated with the vertex cover $C$ is the tree induced by the red edges. The number of leaves of $T$ is $2(|E(G)|+1)-|C|=32$.}
    \label{fig:reduc-planar}
\end{figure}

\begin{lemma}
Every spanning tree of $G'$ has at most $2(|E(G)|+1)-\tau(G)$ leaves.
\end{lemma}

\begin{proof}
Let $k:=2(|E(G)|+1)-\tau(G)$.
Assume by contradiction that $G'$ has a spanning tree $T$ with at least $k+1$ leaves. First, observe that if we require every edge-vertex to be a leaf in $T$, then $T$ has at most $k$ leaves. Indeed, as we already noticed, every face-vertex is an internal node. Then, minimizing the number of original vertices that have to be internal nodes in $T$ is equivalent to minimize the size of a vertex cover in $G$. Hence, the total number of internal nodes in $T$ is at least $|F(G)|+\tau(G)$ and thus the number of leaves is at most $|V(G')|- |F(G)|-\tau(G) = |V(G)|+|E(G)|+|F(G)|-\tau(G)=k$ since $|F(G)|=2-|V(G)|+|E(G)|$ by Euler's formula.

It follows that since $T$ has at least $k+1$ leaves, then $T$ \emph{must} contain an edge-vertex $w_{uv}$ as an internal node. So both $uw_{uv}$ and $vw_{uv}$ are in $T$. Let $T' = T \setminus \{uw_{uv}\}$. We denote by $C_u$ (respectively $C_v$) the connected of $T'$ containing $u$ (respectively $v$). 
By symmetry, we can assume that $w_f \in C_u$. If we add $vw_f$ to $T'$, the resulting set of edges $T''$ induces a spanning tree of $G'$. Besides, the number of leaves in $T''$ is at least the number of least in $T$ since $w_{uv}$ has degree one in $T'$ and $w_f$ was already an internal node in $T$. The number of edge-vertices which are internal nodes have decreased without increasing the number of internal nodes. 
We repeat this process as long as there is at least one internal edge-vertex. We end up with a spanning tree in which every edge-vertex is a leaf and which contains at least $k+1$ leaves, a contradiction. 
\end{proof}

\begin{lemma} \label{lemma:vc-st}
For any minimum vertex cover $C$ of $G=(V,E)$, we can define a canonical tree with exactly $k := 2(|E(G)|+1)-\tau(G)$ leaves which are all the edge-vertices, all the leaf-vertices and all the original vertices but the ones in $C$. Moreover, this spanning can be computed in polynomial time.
\end{lemma}

\begin{proof}
We first explain how to construct $T$ from $C$. 
For every edge-vertex $w_{uv}$, we select in $T$ an edge between $w_{uv}$ and a vertex of $\{u,v \} \cap C$ (if both $u$ and $v$ are in $C$ we attach it to the one with the minimum label value). Such a vertex exists since $C$ is a vertex cover of $G$. For every face $f$, we select the edge $w_fu_f$.

Let $f_\o$ be the outer face and let $w_\o$ be the face-vertex of $f_\o$. We attach every vertex of $f_\o$ to $w_\o$. If the resulting graph is already a spanning tree, we are done.
 
We say that two faces are adjacent if they share a common edge. We now consider the following graph $G''$: we create a vertex for every face of $G$ and two vertices of $G''$ are adjacent if the corresponding faces of $G$ are adjacent. In other words, $G''$ is the dual graph of $G$, without multiple edges. We then run a breadth-first search algorithm from the vertex of $G''$ which corresponds to $f_\o$. Here again, we can first label the vertices in order to process the children in the same order. We use the breadth-first search to incrementally increase the size of the connected component of $T$ which contains $w_\o$ and denoted by $S_\o$. Observe that every vertex of $f_\o$ belongs to $S_\o$.

Now, let $f_i$ be the $i$-th face visited by the breadth-first search traversal. We assume that all the vertices that belong to faces whose index is strictly less than $i$ already belong to $S_\o$. This includes the edge-vertices and the face-vertices with their respective degree-one neighbor. We now explain how to add the vertices of $f_i$ to $S_\o$. Let $f_j$ be the parent of $f_i$ in the BFS traversal, for some $j < i$. By assumption, all the vertices of $f_j$ belong to $S_\o$. Since $f_j$ is the parent of $f_i$, these two faces share at least one edge. Among all the edges incident to both $f_i$ and $f_j$, we pick the one which is covered in $C$ by the vertex with the smallest identifier. We denote by $u$ this vertex. We attach every vertex in $f_i \setminus S_\o$ to the face-vertex $w_i$. Finally, we attach $w_i$ to $u$. 

Therefore, at the end of the BFS traversal, every vertex belongs to $S_\o$. Since at every step, we only attach vertices that did not belong to $S_\o$ before, we do not create any cycle. It follows that the resulting graph is a spanning tree. Besides, it is clear that it can be computed in polynomial time. It remains to prove that the number of leaves is exactly $2(|E(G)|+1)-\tau(G)$. First, recall that for every planar graph $G=(V,E)$, the number of faces of $G$ is precisely $2-|V|+|E|$. Now, let $T$ be the spanning tree obtained by the previous algorithm. We classify the vertices of $G'$ in four different categories: the edge-vertices, the face-vertices, the leaves attached to these face-vertices, and finally the original vertices from $G$. By construction, each edge-vertex and each vertex in $V(G) \setminus C$ is a leaf in $T$. On the other hand, each face-vertex is an internal node in $T$ since it must be adjacent to his degree-one neighbor and it must be connected to the rest of the spanning tree $T$. Finally, since $C$ is minimum and thus minimal, for every vertex $u \in C$, there is an edge $uv \in E(G)$ which is only covered by $u$. Therefore, it follows from the construction of $T$ that the corresponding edge-vertex $w_{uv}$ is attached to $u$ and thus that $u$ is an internal node. 

As a result, the total number of leaves in $T$ is $|F(G)|+|E(G)|+|V(G)|-|C| = 2(|E(G)|+1)-\tau(G)$, as desired. 
\end{proof}

Recall that $(G,C_\source, C_\target)$ is an instance of \textsc{Minimum Vertex Cover Reconfiguration}. We already explained how to construct the corresponding graph $G'$ from $G$. By Lemma \ref{lemma:vc-st}, we can compute in polynomial time two spanning trees $T_\source$ and $T_\target$ from $C_\source$ and $C_\target$ with $2(|E(G)|+1)-\tau(G)$ leaves. Finally, we set $k = 2(|E(G)|+1)-\tau(G))$. Let $(G',k,T_\source,T_\target)$ be the resulting instance of \STWML. We claim that $(G,C_\source, C_\target)$ is a \textsf{yes}-instance if and only $(G',k,T_\source, T_\target)$ is a \textsf{yes}-instance.

\vspace{\baselineskip}

($\Rightarrow$)
Suppose first that $(G,C_\source, C_\target)$ is a \textsf{yes}-instance and let $S=\langle C_1 = C_\source, C_2, \ldots, C_\ell = C_\target\rangle$ be a reconfiguration sequence between $C_\source$ and $C_\target$. For every vertex cover $C_i$ in the sequence, there exists a spanning tree $T_i$ of $G'$ associated with $C$ with $k$ leaves by Lemma \ref{lemma:vc-st}. It is sufficient to show that we can transform two spanning trees $T_i$ and $T_{i+1}$ corresponding to two consecutive vertex covers $C_i$ and $C_{i+1}$, without increasing the number of internal nodes during the transformation. Let $u$ be the vertex of $C_i \setminus C_{i+1}$ and let $v$ be the vertex of $C_{i+1} \setminus C_i$. We first claim that $uv \in E(G)$. Suppose that $uv \not\in E(G)$. Since $v \not\in C_i$, all the neighbors of $v$ belong to $C_i$ by the definition of vertex cover. Therefore, $C_{i+1} \setminus \{ v \}$ contains $N[v]$ and thus is a vertex cover. A contradiction with the minimality of $k$.

Since $v \not\in C_i$, it follows from the construction of $T_i$ that $v$ is a leaf. Therefore, before attaching any vertex to $v$, we first need to reduce the degree of $u$. Since $C_i \, \triangle \, C_{i+1} = \{u,v\}$, we have that $N[u] \setminus \{v\} \subseteq C_i$. Recall that every vertex that belongs to $C_i$ is an internal node in $T_i$. Let $X$ be the set of edge-vertices except $w_{uv}$ attached to $u$ in $T_i$. First, we attach every vertex in $X$ to its other extremity. 

Now, we root $T_i$ and $T_{i+1}$ on the leaf attached to the face-vertex of the outer face, denoted by $w_\o$. If $u$ belongs to the outer face, its parent in $T_i$ and $T_{i+1}$ is $w_\o$. Therefore, for every face $f$ incident to $u$ such that the corresponding face-vertex $w_f$ is attached to $u$ in $T_i$, we attach $w_f$ to the same vertex as in $T_{i+1}$, except if this vertex is $v$. Since we do not want to increase the number of internal nodes, we first need to attach $w_f$ to a vertex in $(f \cap C_i) \setminus \{u\}$. Note that this vertex exists since any vertex cover contains at least two vertices per face. It follows that now $u$ has degree two. Therefore, we can attach the edge-vertex $w_{uv}$ so that $u$ becomes a leaf and $v$ and internal node. Let $T'$ be the resulting tree. Finally, we can now attach to $v$ every face-vertex that is adjacent to it in $T_{i+1}$.

If $u$ does not belong to the outer face, we need to be more careful since we should not isolate $u$ while modifying $T_i$ into $T_{i+1}$. Recall that the parent of $u$ in $T_i$ is the face-vertex corresponding to the first face incident to $u$ visited during the BFS traversal. Since the labeling of the faces is independent of the vertex cover, $u$ has the same parent in $T_{i+1}$ as in $T_i$. The same argument also applies to $v$ and thus the parent of $v$ is the same in $T_i$ and $T_{i+1}$. Therefore, $(G',k,T_\source, T_\target)$ is a \textsf{yes}-instance, as desired.

\vspace{\baselineskip}
($\Leftarrow$)
For the other direction, let $S' = \langle T_1 = T_\source, T_2, \ldots, T_{\ell-1}, T_\ell=T_\target\rangle$ be a reconfiguration sequence between $T_\source$ and $T_\target$ such that the number of leaves is at least $k$ at any time. Recall that the number of leaves in $T_\source$ and $T_\target$ is maximal. Hence, each spanning tree in $S'$ has exactly $k$ leaves.

We claim that every edge-vertex is a leaf in any spanning tree of $S'$. First, recall that this statement holds for $T_\source$ and $T_\target$. Let $T_i$ be the first spanning tree in $S'$ which contains an edge-vertex as an internal node. Since every edge-vertex is a leaf in $T_{i-1}$ and $|E(T_{i-1}) \, \triangle \, E(T_i)| = 2$, exactly one edge-vertex in $T_i$ is an internal node. Let $w_{uv}$ be this vertex. We assume without loss of generality that $uw_{uv} \in E(T_{i-1})$ and thus the edge in $T_i \setminus T_{i-1}$ is $vw_{uv}$. We consider the (only) edge in $T_{i-1} \setminus T_i$, denoted by $e$. $T_{i-1}$ contains three kinds of edges: between an original vertex and an edge-vertex, between an original vertex and a face-vertex, or between a leaf and a face-vertex.
Since all the vertices of the form $u_f$ or $w_{xy}$ have degree one in $T_{i-1}$, $e$ is necessarily of the form $xw_f$, i.e.\ an edge linking a face-vertex and an original vertex. Recall that $w_f$ is an internal node in any spanning tree of $G'$. Since $w_{uv}$ is a leaf in $T_i$ but not in $T_{i+1}$, the degree of $x$ in $T_i$ \emph{must} be two, otherwise we would increase the total number of internal nodes. Note that $uv \in E(G)$ since $w_{uv}$ is an edge-vertex of $G'$ and thus $G'$ contains a face-vertex $w_{f'}$ adjacent to both $u$ and $v$.  Let $T'_{i+1}$ be the forest obtained from $T_{i+1}$ by removing the edge $vw_{uv}$ and observe that $T'_{i+1} = T_i \setminus \{xw_f\}$. We denote by $C_u$ (respectively $C_v$) the connected component of $u$ (respectively $v$) in $T'_{i+1}$. We apply the same argument as in the proof of Lemma \ref{lemma:vc-st}. The node $w_{f'}$ has a neighbor either in $C_u$ or in $C_v$ (which might be $u$ or $v$) but not in both $C_u$ and $C_v$ otherwise $T_{i+1}$ would contain a cycle. We assume without loss of generality that $w_{f'} \in C_u$. Then, observe that if we add the edge  $vw_{f'}$ to $T'_{i+1}$, we get a spanning tree of $G'$ such that $|T_i \, \triangle \, T'_{i+1}| = 2$ but with $k+1$ leaves, a contradiction. 

It follows that for every $T_i \in S'$, $1 \le i \le \ell$, the number of leaves in $T_i$ is exactly $k = 2(|E(G)|+1)-\tau(G)$, and every edge-vertex of $G'$ is a leaf in $T_i$. From $T_i$, we can deduce a vertex cover $C_i$ of $G$: the vertex that covers the edge $uv \in E(G)$ in $C_i$ corresponds to the neighbor of the edge-vertex $w_{uv}$ in $T_i$. In particular, the corresponding vertex covers of $T_\source$ and $T_\target$ are $C_\source$ and $C_\target$, respectively.

Then, from $S'$, we can deduce a sequence $S'' = \langle C_1 = C_\source, C_2, \ldots, C_{\ell'}= C_\target\rangle$ of vertex covers of $G$. Note that the length of $S''$ is not necessarily the same as the length of $S'$, i.e.\ it is possible that two adjacent spanning trees $T_i$ and $T_{i+1}$ in $S'$ give the same corresponding vertex cover of $G$. It remains to prove that $|C_i| = \tau(G)$ for every $1 \le i \le \ell'$, and $|C_i \, \triangle \, C_{i+1}|=2$ for any two adjacent vertex covers of $S''$, i.e.\ $S''$ is a {\sf TJ}-sequence of minimum vertex covers of $G$. Since $|C_1|=|C_{\ell'}|=\tau(G)$, it is sufficient to prove that $|C_i \, \triangle \, C_{i+1}| = 2$, for every $1 \le i < \ell'$. Let $C_i$ and $C_{i+1}$ be two consecutive vertex covers in $S''$. Let $i'$ be the maximal index such that the vertex cover induced by the spanning tree $T_{i'} \in S'$ is $C_i$. Due to the maximality of $i'$, the vertex cover induced by $T_{i'+1}$ corresponds to $C_{i+1}$, since it cannot be $C_i$. Therefore, the edge in $T_{i'} \setminus T_{i'+1}$ is between an edge vertex and an original vertex. We denote by $uw_{uv} \in E(G')$ this edge. Then, since $w_{uv}$ has degree one in $T_{i'}$, the edge in $T_{i'+1} \setminus T_{i'}$ must be $vw_{uv}$. Therefore, $C_{i+1} = (C_i \setminus \{u\}) \cup \{v\}$ and thus $|C_i \, \triangle \, C_{i+1}| \le 2$ holds, for every $1 \le i < \ell'$ as desired. Hence, $(G',k, T_\source, T_\target)$ is a \textsf{yes}-instance. This concludes the proof of Theorem \ref{thm:stwml-planar}.

\subsection{Two internal nodes}

\begin{theorem}\label{thm:stwml_2int}
 Let $G$ be a graph and $T_\source$ or $T_\target$ be two spanning trees with at most two internal nodes.  Then we can check in polynomial time if one can transform the other via a sequence of spanning trees with at most two internal nodes.
\end{theorem}

\begin{proof}
 We first consider the case where either $T_\source$ or $T_\target$ has one internal node, but not both. We assume without loss of generality that $in(T_\source) = \{u\}$, with $u \in A$. If $u \in in(T_\target)$, we just have to attach every leaf in $T_\target$ to $u$, one by one. It follows that $in(T_\source) \cap in(T_\target) = \emptyset$ and thus $u$ has degree one in $T_\target$. Hence, if we want to reconfigure $T_\source$ into $T_\target$, we \emph{must} remove all but one edges incident to $u$ and thus we must create a new internal node. Therefore, it is sufficient to consider the last following case: $|in(T_\source)| = |in(T_\target)| = 2$.
    
First, observe that if $in(T_\source) = in(T_\target)$, then $(G,k,T_\source,T_\target)$ is a \textsf{yes}-instance. Indeed, we just have to change the parent of a node, and this can be done without increasing the number of internal nodes. Hence, in the remaining of the proof of this case, we only consider the case $in(T_\source) \neq in(T_\target)$. 

A vertex $u$ is a \emph{pivot} vertex of $G$ if  $\deg{u} \ge n-2$ in $G$ ($\deg{u}$ being the size of the neighborhood of $u$, $u$ not included). A spanning tree $T$ of $G$ is {\em frozen} if all the spanning trees in its component of the reconfiguration graph have the same internal nodes. 

\begin{claim} \label{claim:frozen}
Let $T$ be a spanning tree of $G$. If $in(T)$ does not contain a pivot vertex, then $T$ is frozen.
\end{claim}
    
\begin{proofclaim}
By contradiction. Assume that $in(T)$ does not contain a pivot vertex and thus each vertex in $in(T)$ has degree at most $n-3$. Then, we want to prove that we cannot modify $in(T)$. Let $in(T) = \{u,v\}$, and note that $uv \in E(T)$. Note also that since $\deg{u} \le n-3$ and $\deg{v} \le n-3$, both $u$ and $v$ have degree at least three in $T$. 
Since $k=2$ and $|in(T)|=2$, we first need to lower the degree of $u$ or $v$ to one or two, without creating a new internal node. Suppose without loss of generality that we want to lower the degree of $u$, the other case follows by symmetry. 
First, observe that we cannot remove the edge $uv \in E(T)$ with an edge flip because it would create a new internal node, as the degree of both $u$ and $v$ is at least three. Recall that $\sum_{u \in V(T)} \deg_T{u} = 2n -2$. Since $T$ has $n-2$ leaves, $\deg_T{u} + \deg_T{v} = n$. Hence, if we want $u$ to have degree two, $v$ must have degree $n-2$, which is not possible. \claimqed
\end{proofclaim}

\begin{claim}\label{claim:pivot}
Let $u$ be a pivot vertex. All the trees containing $u$ as internal vertex are in the same connected component of the reconfiguration graph.
\end{claim}

\begin{proofclaim}
Let $T$ and $T'$ be two trees such that $u \in in(T) \cap in(T')$. If the other internal vertices (if they exist) are the same, then the conclusion follows from Lemma~\ref{lem:stwml_samecomp}. So we can assume that $in(T)=\{u,v\}$ and $in(T')=\{u,w\}$ with $v \ne w$. 
Since $\deg{u} \ge n-2$, there exists a spanning tree $T_2$ with internal nodes $\{u,v\}$ such that $\deg{u}=n-2$ and $\deg{v}=2$ and $uv \in T_2$. By Lemma~\ref{lem:stwml_samecomp}, this spanning tree is in the component of $T$. Let $z$ be the neighbor of $v$ distinct from $u$. Now remove the edge $vz$ and create $wz$ or $uz$ (one of them must exist since $\{u,w\} = in(T_2)$. The internal nodes of the resulting tree is in $\{u,w \}$ and then the conclusion follows by Lemma~\ref{lem:stwml_samecomp}. \claimqed
\end{proofclaim}

A spanning tree $T$ \emph{contains a pivot vertex} if $in(T)$ contains a pivot vertex.
By Claim \ref{claim:pivot}, if $T_\source$ and $T_\target$ contains a common pivot vertex, then the answer is positive. (Note that the existence of a pivot vertex can be checked in polynomial time). If $T_\source$ or $T_\target$ does not contain any pivot vertex, then the answer is negative by Claim~\ref{claim:frozen} (except if the set of internal nodes are the same by Lemma~\ref{lem:stwml_samecomp}). So we restrict our attention to the case where they contain a pivot vertex which is different.

Let $in(T_\source)= \{u,v \}$ and $in(T_\target) = \{x,y\}$ where $u$ and $x$ are pivot vertices. (note that we can possibly have $v=y$). If $u$ (or $x$) is a universal vertex, we can modify $in(T_\source)$ (or $in(T_\target)$) into a spanning tree $T$ with $in(T)=\{u \}$ (resp. $\{x \}$). Claim~\ref{claim:pivot} ensures that both $T_\source$ and a spanning containing $u$ and $x$ as internal nodes are in the same component. And this latter spanning tree is in the component of $T_\target$ by Claim~\ref{claim:pivot}. So we can assume that none of the four internal vertices is universal.

If $in(T_\source)$ or $in(T_\target)$ contains two pivot vertices, w.l.o.g. $in(T_\source)$, $u \cup x$ or $v \cup x$ dominates $G$. So there exists a spanning tree $T$ with $in(T)=\{ u,x \}$ or $\{v,x \}$. Up to symmetry, let us say $\{u,x \}$. Again by Claim~\ref{claim:pivot}, $T$ is both in the connected component of $T_\source$ and $T_\target$.

So $in(T_\source)$ and $in(T_\target)$ contain exactly one pivot vertex; respectively $u$ and $x$. 
Observe that, if we want to reconfigure $T_\source$ into $T_\target$, we must remove $u$ from the spanning tree at some point since it does not belong to $in(T_\target)$. But then, just before disappearing, the second internal node has to have degree $n-2$ in the spanning tree, and then has to be a pivot vertex. So the previous paragraph ensures that $T_\source$ can be transformed into $T_\target$ if and only if there exists a spanning tree in the component of $T_\source$ with two pivot vertices. It is the case if and only if there exists a second pivot vertex $w$ such that $\{u,w \}$ dominates the graph, which can be checked in polynomial time.
\end{proof}

One can naturally wonder if this can be extended to larger values of $k$ or if it is special for $k=2$. We left this as an open problem. We were only interested in the case $k=2$ since it was of particular interest for cographs.

\subsection{Cographs}

Recall that the family of cographs can be defined as the family of graphs with no induced $P_4$, or equivalently by the following recursive definition:

\begin{itemize}
    \item $K_1$ is a cograph;
    \item for $G_1$ and $G_2$ any two cographs, the \emph{disjoint union} $G_1 \cup G_2$ is a cograph (the disjoint union being the graph with vertex set $V(G_1) \cup V(G_2)$ and edge set $E(G_1) \cup E(G_2)$);
    \item for $G_1$ and $G_2$ any two cographs, the \emph{join} $G_1 + G_2$ is a cograph (the join being the graph with vertex set $V(G_1) \cup V(G_2$ and edge set $E(G_1) \cup E(G_2) \cup V(G_1)\times V(G_2)$).
\end{itemize}

Cographs can be recognized in linear time, see e.g. \cite{HABIB2005183}.

\begin{theorem}\label{thm:stwml-cographs}
 \STWML can be decided in polynomial time on cographs.
\end{theorem}

\begin{proof}
Let $G=(V,E)$ be a cograph and let $(G,k, T_\source, T_\target)$ be an instance of \STWML. We denote by $n$ the number of vertices of $G$. First, since $T_\source$ and $T_\target$ are two spanning trees of $G$, $G$ must be connected. Hence, $G$ has been obtained from the join of two graphs, let us say $A$ and $B$. Recall that maximizing the number of leaves of a spanning tree is equivalent to minimizing the number of internal nodes. Hence, in the remaining of this proof, we refer as $k$ to the threshold on the maximum number of internal nodes.

If $k=1$, any spanning tree of $G$ is a star and thus contains exactly one internal node. Therefore, two spanning trees of $G$ are reconfigurable if and only if $G$ contains at most three vertices or the same internal vertex by Lemma~\ref{lem:stwml_samecomp}. Hence, we can safely assume that $k \ge 2$. Since G is the join of two cographs, $G$ can be partitioned into two subsets $A$ and $B$ such that  $G[A]$ and $G[B]$ are two cographs, and we have all possible edges between $A$ and $B$. Let $T$ be a spanning tree of $G$, and let $in(T)$ be its set of internal nodes. We say that $T$ is an \emph{$A$-tree} (resp. \emph{$B$-tree}) if $in(T) \subseteq A$ (resp. $in(T) \subseteq B$. Otherwise, we say that $T$ is an \emph{$(A,B)$-tree}.

If $k=2$, Theorem~\ref{thm:stwml_2int} ensures that the problem can be decided in polynomial time. So from now on, we can assume that $k \ge 3$.
In this case, we claim that $(G,k,T_\source,T_\target)$ is a \textsf{yes}-instance. 

Let us first prove by induction on the size of $G$ that there exists a transformation from any tree $T$ with at most $k$ internal nodes to a tree $T'$ with at most 
$k-1$ of them such that all along the transformation there exists a vertex $x$ which is always an internal node. We moreover prove that this transformation can be found in polynomial time. If $T$ has at most two internal nodes, the conclusion follows. So we can assume that $T$ has exactly $k$ internal nodes.

 If $T$ is a $(A,B)$-tree, we can reach $T'$ as follows. Let $a \in in(T) \cap A$ and $b \in in(T) \cap B$ such that $ab \in T$ (such an edge must exist). Using edge flips, we make $a$ adjacent to any vertex in $B$ and $b$ incident to every vertex of $A$ (which is possible since $A-B$ is a join). After all these modifications, the resulting tree has exactly two internal nodes.
 
 So we can assume that $T$ is an $A$-tree or a $B$-tree, without loss of generality an $A$-tree. Thus every vertex in $B$ is a leaf and then the restriction $T_A$ of $T$ to $G[A]$ also is a spanning tree of $G[A]$. By induction, since $G[A]$ is a connected cograph, we can find in polynomial time a transformation of $T_A$ into a $T_A'$ in such a way that $x$ is an internal node all along the transformation (and this transformation can be found in polynomial time). This transformation can be adapted for $T$ by first connecting all the vertices of $B$ to $x$ using edge flips and then transforming the edges of $G[A] \cap T$ into $T_A'$. All along the transformation $x$ is an internal node and at any step the set of internal nodes are precisely the ones of the tree restricted to $G[A]$.
 
So we can assume that $T_\source$ and $T_\target$ have at most $k-1$ internal nodes.
Let us define a {\em canonical} spanning tree $T_\canonical$ with two internal nodes and show that both $T_\source$ and $T_\target$ can be reconfigured into $T_\canonical$. We define $in(T_\canonical)$ as follows: we pick a vertex $a \in A$, and a vertex $b \in B$ arbitrarily. We only explain without loss of generality how to reconfigure $T_\source$ into $T_\canonical$. 

Since $|in(T_\source)|<k$, we can trivially modify it into $in(T_\canonical)$.
We only show the statement for $|in(T_\source)|=2$, and $k=3$. The proof is similar for other values of $k$. Let $in(T_\source) = \{u,v\}$. Suppose first that $T_\source$ is an $A$-tree or a $B$-tree. We will consider the case where $T_\source$ is an $(A,B)$-tree later. We assume without loss of generality that $T_\source$ is an $A$-tree. We first add $b$ to $in(T_\canonical)$, i.e. we attach each vertex in $A$ to $b$. Observe that we can now remove a vertex in $\{u,v\}$ since all the vertices in $A$ are covered by $b$ and only vertex is needed to cover $B$. It follows that $T_\source$ is now an $(A,B)$-tree with two internal vertices. Hence, we can now first $a$ (it creates a third internal node but this is allowed since $k \ge 3$. It remains to remove the vertex in $(in(T_\source) \cap A) \setminus \{a\}$.
This concludes the proof of Theorem \ref{thm:stwml-cographs}.
\end{proof}

\subsection{Interval graphs}

A graph $G$ is an \emph{interval graph} if $G$ can be represented as an intersection of segments on  the line. More formally, each vertex can be represented with a pair $(a,b)$ (where $a\leq b$) and vertices $u=(a,b)$ and $v=(c,d)$ are adjacent if the intervals $(a,b)$ and $(c,d)$ intersect.
Let $u=(a,b)$ be  a vertex; $a$ is the \emph{left extremity} of $u$ and $b$ the \emph{right extremity} of $u$. The left and right extremities of $u$ are denoted by respectively $l(u)$ and $r(u)$. Given an interval graph, a representation of this graph as the intersection of intervals in the plane can be found in $\mathcal{O}(|V|+|E|)$ time (see for instance~\cite{Booth76}). 
Using small perturbations, we can moreover assume that all the intervals start and end at distinct points of the line. In the remaining of this section we assume that we are given such a representation of the interval graph.

\begin{theorem}\label{thm:stwml-interval}
 \STWML can be decided in polynomial time on interval graphs.
\end{theorem}
The proof techniques are inspired from~\cite{BonamyB17}.
The rest of this section is devoted to prove Theorem~\ref{thm:stwml-interval}. Recall that, for every tree, the number of leaves is equal to $n$ minus the number of internal nodes. So, for convenience, our goal would consist of minimizing  the number of internal nodes rather than maximizing the number of leaves.

If $G$ is a clique, then $G$ is a cograph and then the problem can be decided in polynomial by Theorem~\ref{thm:stwml-cographs}. So, from now on, we can assume that $G$ is not a clique and in particular $in(G) \ge 2$.

\paragraph*{Canonical spanning tree.}

Let $G$ be an interval graph (distinct from a clique) given with its representation. All along the proof we assume that the vertices $v_1,\ldots,v_n$ are given by increasing right extremity. The \emph{canonical set} is the subset of vertices returned by the following algorithm
\begin{itemize}
 \item Set $X:= \emptyset$ and $G_0=G$. 
 \item Repeat until the graph $G_i$ is reduced to a clique, add to $X$ the vertex $v_i$ which is the largest vertex such that $v_i$ is incident to $v_j$ for every $j \le i$ in $G_i$. And set $G_{i+1}:=G_i \setminus \{ v_j , j <i\}$ 
 \item Return $X$.
\end{itemize}

\begin{figure}[bt]
    \centering
    \includegraphics[width=0.9\textwidth]{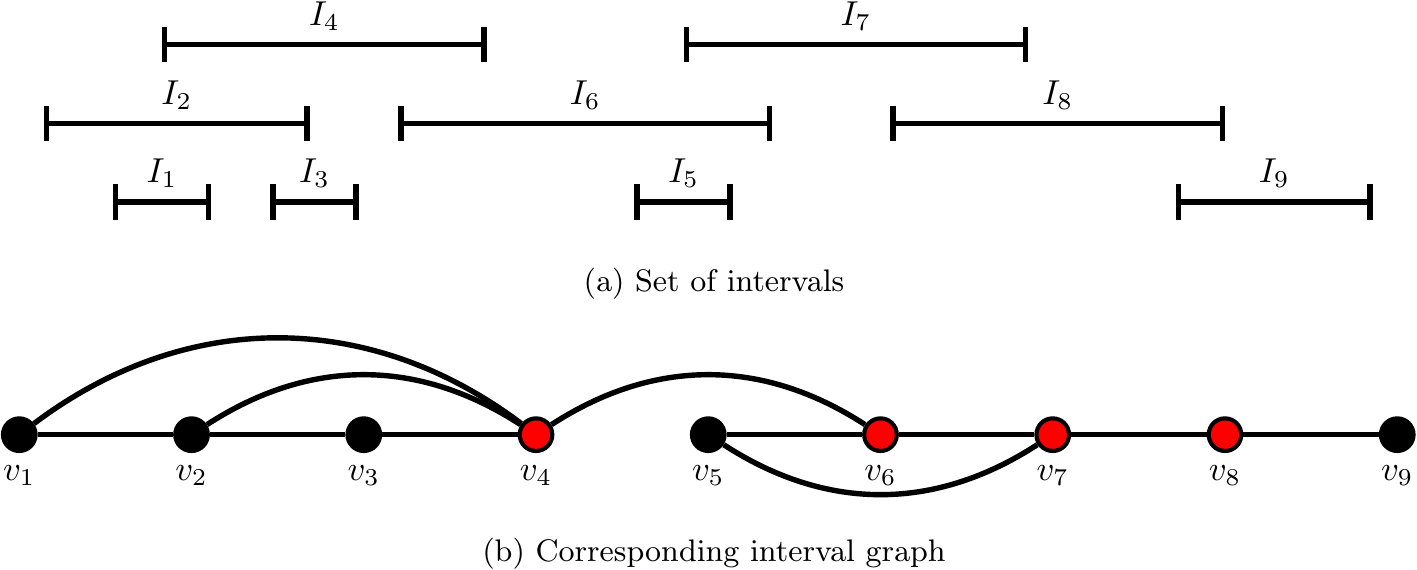}
    \caption{Interval graph with the canonical set $X = \{v_4,v_6,v_7,v_8\}$.}
    \label{fig:interval}
\end{figure}

Note that, after step $i$, the set $X$ is connected and dominates all the vertices before $v_i$ in $G$. Indeed, for every $v_i \in X$, $v_i$ is incident to all the vertices between $v_{i-1}$ and $v_i$ and since $v_{i-1}$ belongs to $G_i$, $v_{i-1}v_{i}$ is an edge. Moreover, the stopping condition ensures that $X$ is a dominating set of $G$. Indeed we stop when $G_i$ is a clique and by assumption $v_{i-1}$ belongs to $G_i$. We moreover claim that $X$ induces a path. Indeed by maximality of $v_{i-1}$, $v_i$ is not incident to $v_{i-2}$. So the set $X$ induces a path. In particular, we have a natural ordering of the vertices of $X$ since the left and right extremity orderings of $X$ \emph{agree} (in other words the orderings of $X$ given by $l(\cdot)$ and $r(\cdot)$ are the same). See Figure~\ref{fig:interval} for an illustration. 

Since all the vertices of $G$ are connected to at least one vertex of $X$, we can easily construct a tree whose set of internal nodes is included in the canonical set. 
Since all the trees with internal nodes included in $X$ are in the same connected component of the reconfiguration graph by Lemma~\ref{lem:stwml_samecomp}, we define, by abuse of notation, the \emph{canonical tree of $G$} as any spanning tree with internal nodes $X$. We denote by $T_C$ the canonical tree of $G$. The internal nodes of $T_C$ are called the \emph{canonical vertices of $G$}. 

\begin{remark}\label{rem:orderspanning}
For every spanning tree $T$ of $G$, if we order the vertices of $in(T)$ by increasing right extremity, the $i$-th vertex of $in(T)$ ends before the $i$-th vertex of $T_C$.

In particular, every spanning tree of $G$ has at least $in(T_C)$ internal nodes. 
\end{remark}

\begin{proof}
By induction on $i$. For the first internal node of $T$, the result holds. Indeed, the first internal node is either the vertex with the smallest right extremity (and the conclusion holds) or a vertex incident to it (since $T$ is spanning). And by definition of $T_C$ the first internal node is the vertex with the rightmost end which is incident to all the vertices that end before it. 

Assume now that the $i$-th internal node $x$ of $T$ ends before the $i$-th internal node $x'$ of $T_C$. Since, by definition of canonical set, the $(i+1)$-vertex of $T_C$ is the vertex $y$ with the maximal right extremitity incident to $x'$ and since $x$ ends before $x'$, the $(i+1)$-th internal node $y$ of $T$ has to be incident to $x$, and then has to end before $y$.
\end{proof}

\paragraph*{C-minimal components.}

Let $k$ be an integer, $G$ be a graph. We denote by $\mathcal{R}(G,k)$ the edge flip reconfiguration graph of the spanning trees of $G$ with at most $k$ internal nodes. 

Let $T,T'$ be two spanning trees with the same set of internal nodes. Lemma~\ref{lem:stwml_samecomp} ensures that $T$ and $T'$ are in the same connected component of $\mathcal{R}(G,k)$. So in what follows, we will often associate a tree $T$ with its set $in(T)$ of internal nodes.

A tree $T$ is \emph{C-minimum} if no tree $T'$ in the connected component of $T$ in $\mathcal{R}(G,k)$ contains fewer internal nodes than $T$. The goal of this part consists of showing that all the trees that are not C-minimum are in the connected component of $T_C$ in $\mathcal{R}(G,k)$. Before doing it, let us give some conditions on the set of internal nodes that ensure that $T$ is not C-minimum:

\begin{lemma}\label{lem:notincluded}
Let $T$ be a spanning tree of $G$ and $k \ge in(T)$. If there exist two internal nodes $u,v$ of $T$ such that the interval of $u$ is included in the interval of $v$ then $T$ is not C-minimum in $\mathcal{R}(G,k)$. Moreover a tree with internal nodes included in $in(T) \setminus \{ u \}$ in the component of $T$ can be found in polynomial time, if it exists.
\end{lemma}

\begin{proof}
First, observe that since the interval $u$ is included in $v$, we have $N_G[u] \subseteq N_G[v]$. Free to add $uv$ and remove any other edge of the cycle created by this addition, we can assume that $uv$ is in $T$.
Let us now prove that we can decrease the degree of $u$ without changing the set of internal nodes while keeping the existence of $uv$ until $u$ becomes a leaf. 
For every vertex $w$ incident to $u$ in $T$ with $w \ne v$, we delete $uw$ from $T$. Since $uv$ is in the tree, $v$ is not in the component of $w$ in $T \setminus uw$. So the edge flip where we remove $uv$ to create $vw$ keeps the connectivity and reduce the degree of $u$, which completes the proof.
\end{proof}

\begin{lemma}\label{lem:notincluded2}
Let $T$ be a spanning tree of $G$. If there exist three pairwise adjacent internal nodes $u,v,w$ such that $N[u] \subseteq N[v] \cup N[w]$ then $T$ is not C-minimum. Moreover a tree with internal nodes included in $in(T) \setminus \{ u \}$ in the connected component of $T$ can be found in polynomial time.
\end{lemma}

\begin{proof}
Free to add $uv$ and remove any other edge of the cycle created by this addition, we can assume that $uv$ is in $T$. Similarly, free to add $vw$ and remove any edge of the cycle created by this addition distinct from $uv$ we can assume that $vw$ is in $T$. We prove that we can decrease the degree of $u$ while keeping the existence of $uv$ and $vw$. If there is a leaf $f$ attached on $u$ in $T$, we delete $uf$ from $T$ and replace it by $vf$ or $wf$ (one of them must exist by assumption).
Assume now that all the neighbors of $u$ are internal. 

Let $x$ be a neighbor of $u$ distinct from $v$. The deletion of $xu$ creates two connected components $C_1,C_2$. Since $vw$ is an edge of $T$, we can assume that they both are in $C_1$. Since $uv$ is in $T$, $u,v,w$ are in the same component. So the deletion of $xu$ and the addition or $xv$ or $xw$ reconnects the graph, which completes the proof.
\end{proof}

Note that if $u,v,w$ induce a triangle, then in particular the conditions of Lemma~\ref{lem:notincluded} or~\ref{lem:notincluded2} holds. Indeed, either one interval is included in another or the interval with the smallest left extremity and the one with the largest right extremity dominates the third one.
So, free to perform some pre-processing operations, we can assume in what remains that the set of internal nodes of a spanning $T$ of $G$ induces a path. Indeed, if an internal node $x$ is incident to three other internal nodes $u,v,w$, then either at least two of them contain the left extremity (or right extremity) of $x$, or one interval is strictly included in the interval of $x$. In the first case there is a triangle and we can apply Lemma~\ref{lem:notincluded} or~\ref{lem:notincluded2}. In the second case, we can apply Lemma~\ref{lem:notincluded}.

\begin{lemma}\label{lem:int_onemargin}
Let $G$ be an interval graph and $k$ be an integer. Any spanning tree $T$ of $G$ satisfying $in(T) < k$ is in the connected component of $T_C$ in $\mathcal{R}(G,k)$.
\end{lemma}

\begin{proof}
Lemmas~\ref{lem:notincluded} and~\ref{lem:notincluded2} ensure that, up to a pre-processing running in polynomial time, we can assume that $in(T)$ induces a path.
If $in(T_C) \subseteq in(T)$, the conclusion holds by Lemma~\ref{lem:stwml_samecomp}. So there exists an internal node of $T_C$ that is not internal in $T$.
Let us order the vertices of $T_C$ according to their right extremity. Let $x$ be the first vertex of $in(T_C) \setminus in(T)$.
Let us prove that we can transform $T$ into a spanning tree $T'$ such that the first vertex of $in(T_C)$ that is not in $in(T')$ is after $x$.
Now, let us modify $T$ using edge flips in order to obtain a tree $T'$ such that $|in(T')| \le |in(T)|$ and all along the transformation, the trees $S$ satisfy $in(S) \subseteq in(T) \cup \{ x \}$.

Let $z$ be the internal node of $T_C$ before $x$ \footnote{If it exists, this vertex might not exist if $z$ is the first vertex of $T_C$.}. Note that, by assumption, $z \in in(T)$. 
Since $in(T)$ is a tree, Remark~\ref{rem:orderspanning} ensures that the internal node $y$ of $T$ that ends after $z$ (or the first one if $z$ does not exist) is adjacent to $x$. Moreover, since $in(T)$ is a tree, $yz$ is an edge. So $x,y,z$ is a triangle and $l(z) \le l(y)$ and $r(y) \le r(x)$. So we have $N(y) \subseteq N(x) \cup N(z)$ \footnote{If $x$ does not exist, we simply have $N(y) \subseteq N(z)$.}. 
And then by Lemma~\ref{lem:notincluded2} we can obtain a tree $T'$ where $x$ is internal and $y$ is not, where the number of internal nodes have not increased and such that the first vertex of $in(T_C)$ that is not in $in(T')$ is after $x$.
\end{proof}

As a direct corollary, we obtain:

\begin{corollary}\label{coro:stwml_notCmin}
Let $G$ be an interval graph and $T$ be a spanning tree with at most $k$ internal nodes. If $T$ is not C-minimum in $\mathcal{R}(G,k)$ or if $T$ has less than $k$ internal nodes, then $T$ is in the connected component of $T_C$ in $\mathcal{R}(G,k)$.
\end{corollary}

\paragraph*{Full access.}
Let $T$ be a tree such that $in(T)$ induces a path. Recall that the left and right extremities orderings agree.
The \emph{leftmost vertex} of $T$ is the vertex of $in(T)$ that is minimal for both $l$ and $r$. The \emph{$i$-th internal node} of $T$ is the internal node with the $i$-th smallest left extremity. 

Let $G$ be an interval graph and $v \in V(G)$. The \emph{auxiliary graph} $H_v$ of $G$ on $v$ is defined as follows. The vertex set of $H_v$ is $v$ plus the set $W$ of vertices $w$ which end after $v$ and start after the beginning of $v$ (i.e. vertices whose interval ends after $v$ but does not contain $v$) plus a new vertex $x$, called the \emph{artificial vertex}. The set of edges of $H_v$ is the set of edges induced by $G[W \cup \{v \}]$ plus the edge $xv$.

\begin{claim}
Let $G$ be an interval graph and $v$ be a vertex of $G$. The graph $H_v$ is an interval graph.
\end{claim}

\begin{proof}
Let $V'=V(H_v) \setminus \{x\}$. For every $v' \in V'$, the interval of $v'$ in $H_v$ is the one of $G$. Now, since we can assume that no interval start at the same point and by construction $l(v) < l(v')$ for every $v' \ne v$ in $V'$, there exists $\epsilon$ such that $v$ is in only interval intersecting $[l(v),l(v)+\epsilon]$. The interval of $x$ is set to $[l(v)-\epsilon/2,l(v)+\epsilon/2]$. It provides an interval representation of $H_v$.
\end{proof}

Let $v \in V(G)$. Every spanning tree of $H_v$ necessarily contains $v$ in its set of internal nodes. Indeed, by construction, the graph $H_v$ contains a vertex $x$ of degree one which is only incident to $v$. Moreover, $v$ is the leftmost internal node of any spanning tree $T$ of $H_v$.

Let $G$ be an interval graph, $k \in \mathbb{N}$ and $T$ be a spanning tree with internal nodes $I$ such that $|I|=k$. 
Let $v \in V(G)$. 
The \emph{restriction} of a spanning tree $T$ to $H_v$ is any spanning tree of $H_v$ with internal nodes included in $(in(T) \cup \{ v \}) \cap V(H_v)$. We denote by $k'_v$ (or $k'$ when no confusion is possible) the value $|(in(T) \cup \{ v \}) \cap V(H_v)|$. Let $T'$ be the restriction of $T$ to $H_v$ as defined above. We claim that the number of internal nodes of $T'$ is at most $k'$. Indeed all the leaves of $T$ attached to internal nodes after $v$ can still be attached to the same internal nodes in $T'$. For those before $v$, either they are not in the graph $H_v$ or they can be attached to $v$. Note that $|in(T) \cap V(H_v)| = k'-1$ if $v \not\in in(T)$, and $|in(T) \cap V(H_v)| = k'$ otherwise.

The vertex $v$ is \emph{good} if the restriction of $T$ to $H_v$ is not C-minimum in $\mathcal{R}(H_v,k')$. The vertex $v$ is \emph{normal} otherwise. 

Let $v$ be a normal vertex. Recall that $v$ is the leftmost internal node of any spanning tree of $H_v$. Let $C$ be the connected component of the restriction of $T$ to $H_v$ in $\mathcal{R}(H_v,k')$. We denote by $\ell'_v(T)$ the second internal node of a spanning tree of $H_v$ in $C$ that minimizes its left extremity.
Similarly we denote by $r'_v(T)$ the second internal node of a spanning tree of $H_v$ in $C$ that maximizes its right extremity.
When they do not exist\footnote{It is the case if and only if $H_v$ is a clique.}, we set $ \ell'_v(T)=- \infty$ and $r'_v(T)=+ \infty$.

We say that we have \emph{full access to $T$} if, for every vertex $v \in V(G)$, we have a constant time oracle saying if $v$ is good or normal. And if $v$ is normal, we moreover have a constant time access to $\ell'_v(T)$ and $r'_v(T)$. What remains to be proved is that {\em (i)} knowing this information for two spanning trees $T$ and $T'$ is enough to determine if they are in the same connected component of $\mathcal{R}(G,k)$, and that {\em (ii)} this information can be computed in polynomial time.

\paragraph*{Dynamic programming algorithm.}
Let us first state the following useful lemma.

\begin{lemma}\label{lem:HvtoH}
Let $G$ be an interval graph and $k \in \mathbb{N}$. Let $T$ be a spanning tree of $G$ and $v$ be an internal node of $T$. Let $J:=in(T) \cap V(H_v)$ and $k'=|J|$. 
If a tree $T'$ with internal nodes $J$ can be transformed into a tree with internal nodes $K$ in $\mathcal{R}(H_v,k')$ then $T$ can be transformed into a tree with internal nodes $(in(T) \setminus J) \cup K$ in $\mathcal{R}(G,k)$. 

In particular, if $T'$ is not C-minimum in $\mathcal{R}(H_v,k')$ then $T$ is not C-minimum in $\mathcal{R}(G,k)$.
\end{lemma}

\begin{proof}
First recall that $in(T)$ contains $v$, and thus $v \in J$. So, the restriction $T'$ of $T$ to $H_v$ (plus the edge $xv$) is a spanning tree of $H_v$ with set of internal nodes $J$. By Lemma~\ref{lem:stwml_samecomp}, we can assume that all the leaves of $T$ in $V(G) \setminus V(H_v)$ are attached to internal nodes in $(in(T) \setminus J) \cup \{ v \}$ and the other are attached to vertices of $J$. 

Let $T'$ be the tree $T$ restricted to $H_v$ plus the edge $vx$ and $\mathcal{S}$ be an edge flip reconfiguration sequence starting from $T'$ and ending with a tree with internal nodes $K \cup \{ v \}$. For every intermediate tree $T_t$, we denote by $J_t$ the set of internal nodes of $T_t$. 
 We claim that we can perform the same edge flip in $H$ and have, at every step, the set of internal nodes $(I \setminus J) \cup J_t$. The first point is due to the fact that any edge of $H_v$ exists in $H$ (but $xv$ which cannot be modified). The second point comes from the fact that we considered a tree $T$ such that all the leaves of $T$ whose right extremity is before the one of $v$ are attached to internal nodes that end before or equal to $v$. So the degree of each vertex $w$ of $H_v$ (but $v$) in $T'$ is the degree of $w$ in the corresponding tree in $G$. And then the conclusion follows.
 \end{proof}

Let $\mathcal{S}$ be a sequence of edge flip adjacent spanning trees such that the set of internal nodes is an induced path all along the sequence. We say that $\mathcal{S}$ is \emph{$j$-fixed} if the first $j$ internal nodes always are the same all along the sequence. Given a $j$-fixed sequence, the \emph{maximum $(j+1)$-th vertex of $\mathcal{S}$} is the $(j+1)$-th internal node of the spanning trees of $\mathcal{S}$ with the maximum right extremity.

Note that any reconfiguration sequence of $G$ is $0$-fixed and for every $v$, any reconfiguration sequence in $H_v$ is $1$-fixed. We will simply use the following lemma in these two cases.

\begin{lemma}\label{lem:secondvertex}
Let $G$ be an interval graph, $k$ an interval graph and $\mathcal{S} = \langle T_1,\ldots,T_\ell \rangle$ be a reconfiguration sequence in $\mathcal{R}(G,k)$ which is $j$-fixed. Let $w$ be the $(j+1)$-th vertex of $\mathcal{S}$ with the rightmost right extremity. Then there is a reconfiguration sequence between trees with internal nodes $(in(T_1) \cap V(H_w)) \cup \{w \}$ and $(in(T_\ell) \cap V(H_w)) \cup \{ w \}$ in $\mathcal{R}(H_w,k'_w)$. 
\end{lemma}

\begin{proof}
Let $1 \le i \le \ell$ be an integer and $X_i= in(T_i) \cap V(H_w)) \cup \{w \}$. Let us define by $T_i'$ the tree of $H_w$ such that all the edges of $T_i$ existing in $H_w$ are in $T_i'$, all the isolated vertices are attached to $w$ and, if $T_i'$ is disconnected, $w$ is attached to one vertex of $X_i$ in the other connected components of $T_i'$. 

We claim that such a tree exists. Indeed, if $v$ is an isolated vertex in $H_w$, then $v$ was attached to an internal vertex which finishes before $w$. And so $v$ can be attached to $w$. Moreover, if a component not reduced to a single vertex is isolated in the resulting forest, there was an edge between an internal node of that component and a vertex that finishes at the left of $w$. Again, this vertex can be connected to $w$.

One can easily check than any edge flip in $\mathcal{S}$ between $T_i$ and $T_{i+1}$ in $G$ can be indeed  adapted into an edge flip in $H_w$: If both edges edges exist, we perform the edge flip; if a vertex is attached on an internal node at the left of $w$, then we attach it on $w$; if the vertices do not exist $T_i'=T_{i+1}'$ and we do not have any operation to perform.
\end{proof}

Let us now prove that if we have full access to $H_v$ for any $v$ we can determine if $T$ is C-minimum and, if it is, the rightmost possible right extremity of the first internal node of the trees in the connected component of $T$ in $\mathcal{R}(G,k)$. Using a similar method, we will then show in Lemma~\ref{lem:stwml_computingR} that we have full access to $H_v$.

\begin{lemma}\label{lem:to_can}
Let $G$ be an interval graph, $k \in \mathbb{N}$, and $T$ be a spanning tree of $G$ with at most $k$ internal nodes Assuming full access to $T$:
\begin{itemize}
    \item We can decide in polynomial time if $T$ is $C$-minimum in $\mathcal{R}(G,k)$;
    \item If $T$ is $C$-minimum, we can moreover compute in polynomial time the rightmost possible right extremity of the first internal node of a tree in the connected component of $T$ in $\mathcal{R}(G,k)$.
\end{itemize}
\end{lemma}
\begin{proof}
First note that if $T$ does not contain any second internal node, then $G$ is a clique and we can conclude (in particular it is a cograph). So we can assume that $G$ is not a clique.

First note that if $|in(T)| <k$ then $T$ can be transformed into the canonical tree by Corollary~\ref{coro:stwml_notCmin}. So we can assume that $in(T)=\{i_1,\ldots,i_k\}$. If $i_1$ is good, then, by Lemma~\ref{lem:HvtoH}, there is a spanning tree $T'$ in the component of $T$ in $\mathcal{R}(G,k)$ such that $|in(T')|<k$. By Corollary~\ref{coro:stwml_notCmin}, $T$ is not C-minimum.

So we can assume that $i_1$ is normal. Let $i'$ be the vertex $\ell'_{i_1}(T)$ in $H_{i_1}$ and $J$ be the set of internal nodes of a spanning tree containing $i'$ as second vertex in the connected component of the restriction of $T$ in $H_{i_1}$ in $\mathcal{R}(H_{v_i},k')$. By Lemma~\ref{lem:HvtoH}, $T$ contains a spanning tree with internal nodes $J$ in its connected component of $\mathcal{R}(G,k)$ (recall that $i_1$ is an internal node of any tree of $H_{i_1}$). 

If $i'$ is incident to all the neighbors of $i_1$, i.e. $N(i_1) \subseteq N(i')$, we can reduce the number of internal nodes by Lemma~\ref{lem:notincluded} and we can conclude with Corollary~\ref{coro:stwml_notCmin} that $T$ is not C-minimum. 

If $i'$ misses at least two neighbors of $i_1$, we claim that the first internal node of all the spanning trees in the component of $T$ in $\mathcal{R}(G,k)$ is $i_1$ and that $T$ is C-minimum. Assume by contradiction that there exists a spanning tree in the connected component of $T$ in $\mathcal{R}(G,k)$ such that the first internal node is distinct from $i_1$. Let $\mathcal{S}$ be a shortest transformation from $T$ to such a tree $T_b$. By minimality of the sequence, $i_1$ is an internal node of all the trees of $\mathcal{S}$ but $T_b$. Let $T_a$ be the tree before $T_b$ in $\mathcal{S}$. Since after an edge flip $i_1$ becomes a leaf, $i_1$ has degree exactly two in $T_a$. Let $j$ be the second internal node of $T_a$. The third internal node cannot be incident to $i_1$. Indeed, since $i_1$ is the first internal node all along the transformation, the transformation from $T$ to $T_a$ is also a transformation from the restriction of $T$ to the restriction of $T_a$ in $H_{i_1}$. And since $i_1$ is normal, we cannot have $i_1$ the second and the third internal node that is a triangle by Lemma~\ref{lem:notincluded2}. So one of the two neighbors of $i_1$ in $T_a$ is $j$. Since $j$ is the second internal node, $j$ has to see all but at most one neighbor of $i_1$. But since the transformation from $T$ to $T_a$ also holds in $H_{i_1}$, the left extremity of $j$ is at least the one of $i'$ by definition of $i'$ which is $\ell'_{i_1}(T)$. And then, by hypothesis on $i'$, $j$ misses at least two neighbors of $i_1$, a contradiction. So in that case, $i_1$ is normal and the rightmost possible right extremity of the first internal node of a tree in the component of $T$ in $\mathcal{R}(G,k)$ is $i_1$.

Finally assume that $N(i')$ contains $N(i_1)$ but exactly one vertex $y$. Let $T_a$ be a spanning tree of $H_{i_1}$ with second internal node $i'$ in the connected component of the restriction $T'$ of $T$ in $\mathcal{R}(H_{i_1},k')$. Let $\mathcal{S}$ be a shortest sequence from $T'$ to such a tree $T_a$.
By Lemma~\ref{lem:HvtoH}, the transformation from $T'$ to $T_a$ in $H_{i_1}$ can be adapted into a transformation from $T$ to some spanning tree of $G$, also denoted by $T_a$. So, we can assume that the second internal node of $T$ is $i'$ (since $T_a$ and $T$ are in the same connected component of $\mathcal{R}(G,k)$). Now, using edge flips, we can remove all the leaves attached to $i_1$ but the edge $i_1y$, en attach them to $i'$. Note that this is possible since $N(i')$ contains $N(i_1) \setminus \{y\}$. If $i_1$ is incident to at least two internal nodes, then $T$ is not minimal by Lemma~\ref{lem:notincluded2}. So $i_1$ has degree two and then it is incident to $i'$ and $y$. Now let $z$ be the first canonical vertex. By definition of $z$, $z$ is incident to $i_1, y$, and $i'$. Indeed, it is incident to all the vertices that end before it and all the second internal nodes of spanning trees by Remark~\ref{rem:orderspanning}.
So we can remove the edge $i_1y$ to create $zy$ instead. The number of internal nodes does not increase since $i_1$ is now a leaf. The right extremity of the first internal node cannot be larger than the right extremity of $z$ by Remark~\ref{rem:orderspanning}. So, in order to conclude, we simply have to decide whether the spanning tree $T$ is C-minimum or not. Since we have full acccess to $T$, we can decide if $z$ is good. If it is good, the restriction of $T$ in $H_v$ is not C-minimum, and then by Lemma~\ref{lem:HvtoH}, $T$ is not C-minimum. Otherwise, we claim that it is. Assume by contradiction that $T$ is not C-minimum and let $\mathcal{S}$ be a shortest sequence to a spanning $T'$ with fewer less internal nodes. By definition of canonical set, the right extremity of the first internal node is always at the left of the right extremity of $z$. And then by Lemma~\ref{lem:secondvertex}, $\mathcal{S}$ also provides a sequence in $H_z$. And then we have in $H_z$ a sequence to a spanning tree with fewer less internal nodes, a contradiction since $z$ is normal.
\end{proof}

We say that we have \emph{full access to $T$ after $v$} if for every vertex $w \in V(G)$ with $w > v$, we have access in constant time to a table that permits us to know whether $w$ is good or normal. And if $w$ is normal, we also have access to $\ell'_w(T)$ and $r'_w(T)$. Using a proof similar to the one of Lemma~\ref{lem:to_can}, one can prove the following:

\begin{lemma}\label{lem:stwml_computingR}
Let $G$ be an interval graph, $k \in \mathbb{N}$, $v \in V(G)$ and $T$ be a spanning tree of $G$ with at most $k$ internal nodes. 

\begin{itemize}
    \item We can decide in polynomial time if $v$ is good if we have full access to $T$ after $v$. 
    \item If $T$ is $C$-minimum, we can moreover compute $r'_v(T)$ and $\ell'_v(T)$ in polynomial time.
\end{itemize}
\end{lemma}

\begin{proof}
If $H_v$ is a clique, then either $k' \ge 2$ and then the restriction of $T$ to $H_v$ is not C-minimum since there is a spanning tree with only one internal node which can be easily reached. Or $k'=1$ and then $\ell'_v(T)=-\infty$ and $r'_v(T)=+\infty$. So we can assume that $H_v$ is not a clique.
Let $w$ be the second vertex of $in(T)$, which exists since $H_v$ is not a clique.
In what follows, by abuse of notations, we still denote by $T$ the restriction of $T$ to $H_v$.

If $w$ is good, then, by Lemma~\ref{lem:HvtoH}, there is a spanning tree $T'$ in the component of $T$ in $\mathcal{R}(H_v,k')$ such that $|in(T')|<k'$. By Corollary~\ref{coro:stwml_notCmin} we can conclude that $T$ is not C-minimum and then $v$ is good.

So we can assume that $w$ is a normal vertex. Let $i'$ be the vertex $\ell'_w(T)$ in $H_{w}$ and $J$ be the set of internal nodes of a spanning tree containing $i'$ as second vertex in $H_{w}$.  By Lemma~\ref{lem:HvtoH}, $H_v$ contains a spanning tree with internal nodes $\{ v \} \cup J$. 

If $N(i') \cup N(v)$ contains $N(w)$ and  $i'$ and $v$ are adjacent, we can reduce the number of internal nodes by Lemma~\ref{lem:notincluded2} and Corollary~\ref{coro:stwml_notCmin} ensures that $T$ is not C-minimum in $H_v$. So from now on, we will assume that it is not the case. In particular, $v,w,i'$ is an induced path.

If $N(v) \cup N(i')$ misses at least one neighbor of $w$, we claim that all the spanning trees in the component of the restriction of $T$ in $\mathcal{R}(H_v,k')$ contain $w$ as second internal node and that $v$ is normal. 
Indeed, assume by contradiction that there exists a spanning tree in the component of $T$ such that the second internal node is distinct from $w$. Let $\mathcal{S}$ be a shortest transformation to such a tree $T_b$. By minimality of $\mathcal{S}$, the second internal node of all the intermediate trees is $w$. Let $T_a$ be the tree just before $T_b$ in the sequence. Since after an edge flip $w$ becomes a leaf, $w$ has degree exactly two in $T_a$. Moreover, if we denote by $j$ the third internal node, we have $l(j) \ge l(i')$ since the transformation is also a transformation in $H_w$. So $N(v) \cup N(j)$ still miss at least one vertex of $N(w)$ and $v,j$ are not incident. So $w$ has to be incident to $v,j$ (since the tree has to be connected) and to the vertices of $N(w) \setminus (N(v) \cup N(i'))$, a contradicticon since $v$ must have degree two in $T_a$ ($v$ becoming a leaf after an edge flip). So $r'_v(T)=\ell'_v(T)=w$.

Finally assume that $N(v) \cup N(i')$ contains $N(w)$ but $vi'$ is not an edge. Let $T_a$ be a tree of $H_{w}$ with second internal node $i'$ in the connected component of the restriction of $T$ in $\mathcal{R}(H_w,k'_w)$. By Lemma~\ref{lem:HvtoH}, the transformation from $T$ to $T_a$ in $H_{w}$ can be adapted into a transformation from $T$ to $T_a$ in $H_v$. So, from now on, we will assume that the third internal node of $T$ is $i'$. Using edge flips, we can remove all the leaves attached to $w$ and attach them to $v$ or $i'$ instead. After these edge flips $w$ is only adjacent to internal nodes. Note moreover that since the restriction of $T$ is C-minimum in $H_w$, no internal node of $T$ in $H_v$ is incident to $w$ but $i'$ and $v$ otherwise we would have been able to apply Lemma~\ref{lem:notincluded2} to the restriction of $T$ in $H_w$. So $w$ has degree at most two in $T$ and is incident to both $v$ and $i'$.
Now let $z$ be the second internal node of the canonical tree of $H_v$ (the first one being necessarily $v$). We claim that $z$ is incident to $w,v$, and $i'$. Indeed, the second internal node has to be incident to the first one, namely $v$. Moreover, since both the intervals of $z$ and $w$ contain the right extremity of $v$, they are adjacent. Since the right extremity of $z$ is after the right extremity of $w$ by Remark~\ref{rem:orderspanning}, $zi'$ is an edge.
So we can remove the edge $wv$ to create the edge among $zv$ and $zi'$ that keeps the connectivity of the graph (since $vwi'$ was a $P_3$, the deletion of $vw$ disconnects them so one of the two edges reconnects the graph). The number of internal nodes does not increase since $w$ is now a leaf. The right extremity of the first internal node $v$ cannot be larger than the right extremity of $z$ by Lemma~\ref{lem:notincluded2}. So $r'_v(T)=z$ if $T$ is C-minimum. Similarly, if we denote by $z'$ the vertex with the smallest left extremity in $H_v$ (distinct from $v$ and the artificial vertex), then we can obtain a spanning tree whose second vertex is $z'$ and then $\ell'_v(T)=z'$ if $T$ is C-minimum.

So, in order to conclude, we simply have to decide if the spanning is C-minimum. Since we have full acccess to $T$, we can decide if $z$ is good. If it is good, the restriction of $T$ in $H_v$ is not C-minimum, and then by Lemma~\ref{lem:HvtoH}, $T$ is not C-minimum. Otherwise, we claim that it is. Indeed, if there is a transformation from $T$ to $T'$ with fewer less internal nodes in $\mathcal{R}(H_v,k'_v)$, all along the transformation $\mathcal{S}$, the second internal node ends before (or is equal to) $z$ by Remark~\ref{rem:orderspanning}. But then by Lemma~\ref{lem:secondvertex}, the transformation $\mathcal{S}$ can be adapted in $\mathcal{R}(H_v,k'_w)$ that decreases the number of internal nodes, a contradiction since $z$ is normal. 
\end{proof}

Lemmas~\ref{lem:stwml_computingR} ensures that we can, using backward induction on the ordering of the vertices, decide in polynomial time for all the vertices $v$ of the graph if a vertex is good and if not we can compute $r'_v(T)$ and $\ell'_v(T)$. So we have full access to $T$ in polynomial time.

\begin{lemma}\label{lem:cclH_v}
Let $G$ be an interval graph and $v$ be a vertex of $G$. Let $T_1,T_2$ be two spanning trees of $G$ with internal nodes $I_1$ and $I_2$ of $H_v$ such that $v$ is normal for both $T_1$ and $T_2$. Let $i_1:=r'_v(I_1)$ and $i_2:=r'_v(I_2)$. The trees $T_1$ and $T_2$ are in the same connected component of $H_v$ if and only if:
\begin{itemize}
    \item $i_1=i_2$ and,
    \item Any spanning trees with internal nodes $(I_1 \setminus \{v\}) \cup \{i_1\}$ and $(I_2 \setminus \{v\}) \cup \{i_2\}$ are in the same connected component of $\mathcal{R}(H_{i_1},k)$.
\end{itemize}
\end{lemma}

\begin{proof}
Let us denote by $g_1$ and $g_2$ the second internal nodes of $T_1$ and $T_2$ respectively.

\noindent ($\Rightarrow$) If $T_1$ can be transformed into $T_2$, then indeed, we must have $i_1=i_2$. And by Lemma~\ref{lem:secondvertex}, the transformation from $T_1$ to $T_2$ is also a transformation from the restriction of $T_1$ into the restriction of $T_2$ in $H_{i_1}$ since the right extremity of the first internal node of all the spanning trees in the sequence is at the left of the one of $i_1$ by definition of $i_1$.

\noindent ($\Leftarrow$) Assume now that both points are satisfied. Let $T_j'$ be a tree with first internal node $i_j$ in the connected component of $T_j$ in $\mathcal{R}(G,k)$ for $1 \le j \le 2$. By Lemma~\ref{lem:secondvertex}, the restriction of $T_j'$ is in the connected component of $T_j$ in $\mathcal{R}(H_{i_j},k)$. Moreover, by assumption there is a transformation from the restriction of $T_1$ to the restriction of $T_2$ in $H_{i_1}$. So by Lemma~\ref{lem:HvtoH}, this transformation can be adapted into a transformation from $T_1$ to $T_2$ in $G$, which completes the proof.
\end{proof}

We now have all the ingredients to prove Theorem~\ref{thm:stwml-interval}.

\begin{proof}[Proof of Theorem~\ref{thm:stwml-interval}]
We can determine in polynomial time if the spanning trees are C-minimum by Lemma~\ref{lem:to_can}. If both of them are not, then both of them can be reconfigured to $T_C$ and there exists a transformation from $T_1$ to $T_2$ by Lemma~\ref{lem:to_can}. If only one of them is, say $T_1$, we can replace $T_1$ by $T_C$ (since they are in the same connected component in the reconfiguration graph). So we can assume that $T_1$ and $T_2$ are C-minimum. And the conclusion follows by Lemma~\ref{lem:cclH_v}.
\end{proof}

\subsection{Still open -- Outerplanar graphs}

There are two types of outerplanar graphs where it is not possible to find a transformation.
\begin{itemize}
    \item $C_4$ plus an edge.
    \item Two paths where we put parallel edges except between the first and the last vertices of the paths. Note that in this case, the construction can be glued together.
\end{itemize}

Questions:
\begin{itemize}
    \item Are they the only obstructions?
    \item Is it always possible to find a transformation when we have a surplus of one?
\end{itemize}

\bibliographystyle{plain}
\bibliography{bibliography}

\end{document}